\documentclass[journal]{IEEEtran}
\usepackage{epsfig}
\usepackage{cite}
\usepackage{graphicx}
\graphicspath{{Figures/}}
\usepackage{url}
\usepackage{stfloats}
\usepackage{amsmath,accents}
\usepackage{amsfonts}   
\usepackage{amsopn}     
\usepackage{psfrag}
\usepackage[tight]{subfigure}
\usepackage{mdwtab}     
\usepackage{mdwlist}
\usepackage[normalem]{ulem}
\usepackage{enumerate}
\usepackage{color}

\usepackage{amsmath}
\usepackage{amsthm}
\usepackage{epsfig}
\usepackage{cite}
\usepackage{graphicx}
\graphicspath{{Figures/}}
\usepackage{subfigure}
\usepackage{url}
\usepackage{stfloats}
\usepackage{amsmath}
\usepackage{amsfonts} 
\usepackage{amsopn}    
\usepackage{psfrag}
\usepackage{xcolor}
\usepackage{subfigure}
\usepackage{mdwtab} 
\usepackage{epstopdf}

\newtheorem{theorem}{\bf Theorem}
\newtheorem{rem}{\bf Remark}
\newtheorem{corollary}{\bf Corollary}
\newtheorem{lemma}{\bf Lemma}
\newtheorem{definition}{\bf Definition}
\newtheorem{prop}{\bf Proposition}

\ifCLASSINFOpdf
\else
\fi
\begin{document}
%
\title{Ultimate Boundedness for Switched Systems with Multiple Equilibria Under Disturbances}
%
%
%

\author{Sushant~Veer and~Ioannis~Poulakakis
        
        
\thanks{S. Veer and I. Poulakakis are with the Department
of Mechanical Engineering, University of Delaware, Newark,
DE, 19716 USA e-mail: {\tt\small \{veer, poulakas\}@udel.edu.}}
\thanks{This work is supported in part by NSF CAREER Award IIS-1350721 and by NRI-1327614.}
}

\maketitle

\begin{abstract}
In this paper, we investigate the robustness to external disturbances of switched discrete and continuous systems with multiple equilibria. It is shown that if each subsystem of the switched system is Input-to-State Stable (ISS), then under switching signals that satisfy an average dwell-time bound, the solutions are ultimately bounded within a compact set. Furthermore, the size of this set varies monotonically with the supremum norm of the disturbance signal. It is observed that when the subsystems share a common equilibrium, ISS is recovered for solutions of the corresponding switched system; hence, the results in this paper are a natural generalization of classical results in switched systems that exhibit a common equilibrium. Additionally, we provide a method to analytically compute the average dwell time if each subsystem possesses a quadratic ISS-Lyapunov function. Our motivation for studying this class of switched systems arises from certain motion planning problems in robotics, where primitive motions, each corresponding to an equilibrium point of a dynamical system, must be composed to realize a task. However, the results are relevant to a much broader class of applications, in which composition of different modes of behavior is required.
\end{abstract}


%
\IEEEpeerreviewmaketitle

\IEEEpeerreviewmaketitle
\section{Introduction}
\label{sec:intro}
A switched system is characterized by a family of dynamical systems wherein only one member is active at a time, as governed by a switching signal. From the perspective of control synthesis, switched systems allow stitching individual controllers under a single framework by viewing the dynamics produced by each controller as an individual system. This gives rise to a convenient and modular control strategy that allows the use of pre-designed controllers for generating behaviors richer than what an individual controller is capable of. Owing to these factors, switched systems have been widely used in a broad range of applications---such as power electronics \cite{vasca2012dynamics}, automotive control \cite{johansen2003gain}, robotics \cite{aguiar2007trajectory}, and air traffic control \cite{tomlin1996hybrid}.

A significant amount of research has been directed towards the stability and robustness of switched systems. Stability of switched linear systems was studied in \cite{narendra1994common} by the construction of a common Lyapunov function which decreases monotonically despite switching. In the absence of a common Lyapunov function, the notion of multiple Lyapunov functions that are allowed to increase intermittently as long as there is an overall reduction, was proposed in \cite{branicky1998multiple}. Instead of dealing with the construction of special classes of Lyapunov functions, \cite{hespanha1999stability} proved that the stability properties of the individual subsystems can be translated to the switched system when switching is sufficiently slow in the sense that the switching signal satisfies an average dwell-time bound. The notion of average dwell time was further exploited in \cite{vu2007input} to study the input-to-state stability (ISS) of continuous switched systems. Detailed surveys of results in switched systems can be found in \cite{lin2009stability} and \cite{liberzon2003switching}; it is emphasized, however, that the aforementioned papers consider switching among systems that share a common equilibrium, as does the majority of the switched systems literature. 

Various applications demand switching among systems that \emph{do not} share a common equilibrium---such as planning motions of legged \cite{motahar2016composing,gregg-planning2012} and aerial \cite{dorothy2016switched} robots, cooperative manipulation among multiple robotic arms \cite{figueredo2014switching}, power control in multi-cell wireless networks \cite{alpcan2004hybrid}, and models for non-spiking of neurons \cite{makarenkov2017dwell}. Such systems are referred to in the literature as \emph{switched systems with multiple equilibria}.\footnote{To clarify terminology, ``switched systems with multiple equilibria'' refers to switching among subsystems each of which exhibits a unique equilibrium which may not coincide with the equilibrium of another subsystem.} To study the behavior of these systems, \cite{alpcan2004hybrid} and \cite{basar2010} established boundedness of the state for switching signals that satisfy an average dwell-time and a dwell-time bound, respectively. The notion of modal dwell-time was introduced in \cite{blanchini2010modal}, which provided switch-dependent dwell-time bounds, while \cite{kuiava2013practical} established boundedness of solutions via practical stability. The dwell-time bound of \cite{basar2010} was extended to discrete switched systems in \cite{motahar2016composing} and to continuous switched systems with invariant sets in \cite{dorothy2016switched}. Yet, papers that deal with multiple equilibria/invariant sets do not study the effect of switching in the presence of disturbances. Conversely, work that consider switching under disturbances is restricted to systems that share a common equilibrium. 
In the present paper, we address this gap in the literature by studying discrete and continuous switched systems with multiple equilibria under disturbances.

Our interest in switched systems with multiple equilibria stems from their application in certain motion planning and control problems in robotics that require switching among different modes of behavior~\cite{burridge1999sequential, tedrake2010}. As a concrete example, consider dynamically-stable legged robots, in which the ability to switch among a collection of limit-cycle gait primitives enriches the repertoire of robot behaviors. This ability provides the additional flexibility needed for navigating amidst obstacles~\cite{gregg-planning2012, motahar2016composing}, realizing gait transitions~\cite{QuCao2015, Cao2016ALR}, adapting to external commands~\cite{veer2017supervisory, quan2016dynamic}, or achieving robustness to disturbances~\cite{saglam2013switching, da2016first}. In this case, each limit-cycle gait primitive corresponds to a distinct equilibrium point of a discrete dynamical system that arises from the corresponding Poincar\'e map---or forced Poincar\'e map~\cite{veer2017poincare}  if disturbances are present. Hence, composing gait primitives can be formulated as a switched discrete system with multiple equilibria, as in~\cite{gregg-planning2012, motahar2016composing, veer2017driftless, veer2017continuum, bhounsule2018switching}. The present paper provides the theoretical tools necessary for ensuring robustness for such systems. It is worth mentioning that these tools can be applied to ensure robust motion planning via the composition of multiple (distinct) equilibrium behaviors in the context of other classes of dynamically-moving robots as well---examples include aerial robots with fixed~\cite{majumdar2017funnel} or flapping~\cite{ramezani2017describing} wings, snake robots \cite{liljeback2011controllability}, and ballbots \cite{nagarajan2014ballbot}.

This paper studies the effect of disturbances on switched discrete and continuous systems with multiple equilibria, where each subsystem is ISS. It is shown that if the switching signal satisfies an average dwell-time constraint---which is analytically computable in the case of quadratic Lyapunov functions---then, the solutions of the switched system are ultimately bounded within a compact set. Furthermore, the size of this set grows monotonically with the supremum norm of the disturbance. With respect to prior literature on switched systems with multiple equilibria, e.g.,~\cite{basar2010, blanchini2010modal, kuiava2013practical, motahar2016composing, dorothy2016switched, makarenkov2017dwell}, the results of this paper explicitly consider the effect of disturbances on the solutions of the switched system. In addition, they constitute a natural generalization of known results in the common equilibrium case; indeed, when the equilibria of the individual subsystems coalesce, ISS of the switched system can be recovered as a simple consequence of our main results, Theorems~\ref{thm:ISS-switch-discrete} and~\ref{thm:ISS-switch-continuous}; see Corollaries~\ref{cor:single-equ-disc} and~\ref{cor:single-equ-cont}.

The paper is organized as follows. Section~\ref{sec:GUUB-def} presents the main results for continuous and discrete switched systems with multiple equilibria. Section~\ref{sec:discrete_switched_system} presents explicit expressions for dwell-time bounds and the relevant set constructions that will be used in the proofs of the main results; the proofs are presented in Section~\ref{sec:proofs}. Section~\ref{sec:examples} considers important implementation aspects and presents numerical examples that illustrate the behaviors of interest in switched discrete and continuous systems. Section~\ref{sec:conclusions} concludes the paper.

\subsection*{Notation} $\mathbb{R}$ and $\mathbb{Z}$ denote the real and integer numbers, and $\mathbb{R}_+$, $\mathbb{Z}_+$ the non-negative reals and integers, respectively. The Euclidean norm is denoted by $\|\cdot\|$ and $B_\delta(x)\subset \mathbb{R}^n$ denotes an open-ball (Euclidean) of radius $\delta>0$ centered at $x\in\mathbb{R}^n$. Let $\mathcal{A}\subseteq \mathbb{R}^n$, then $\accentset{\circ}{\mathcal{A}}$ denotes the interior while $\overline{\mathcal{A}}$ denotes the closure of $\mathcal{A}$, respectively. 
The index $k\in\mathbb{Z}_+$ represents discrete time. The discrete-time disturbance $d: \mathbb{Z}_+ \to \mathbb{R}^m$ is a sequence $\{d_k\}_{k\in\mathbb{Z}_+}$ with $d_k\in\mathbb{R}^m$ for $k\in\mathbb{Z}_+$. The norm of $d$ is $\|d\|_\infty:= \sup_{k\in\mathbb{Z}_+}\|d_k\|$.
Let $t\in\mathbb{R}_+$ represent the continuous time. The disturbance $d:\mathbb{R}_+\to \mathbb{R}^m$ that acts in continuous time is assumed to be a piecewise continuous signal with norm $\|d\|_\infty := \sup_{t\geq 0} \|d(t)\|$. Abusing notation, we use $d$, $\|\cdot\|_\infty$ and $\mathcal{D}$ for both discrete- and continuous-time disturbances. No ambiguity arises as it will always be clear from context whether the signal is discrete or continuous.
Finally, a function $\alpha:\mathbb{R}_+ \to \mathbb{R}_+$ is of class $\mathcal{K}_\infty$ if it is continuous, strictly increasing, $\alpha(0)=0$, and $\lim_{s\to\infty} \alpha(s)=\infty$. A function $\beta:\mathbb{R}_+ \times \mathbb{R}_+ \to \mathbb{R}_+$ is of class $\mathcal{KL}$ if it is continuous, $\beta(\cdot,t)$ is of class $\mathcal{K}_\infty$ for any fixed $t\geq 0$, $\beta(s,\cdot)$ is strictly decreasing, and $\lim_{t\to \infty} \beta(s,t)=0$ for any fixed $s\geq 0$; see~\cite{khalil2002nonlinear} for class $\mathcal{K}_\infty$ and $\mathcal{KL}$ functions.

\section{Main Results}
\label{sec:GUUB-def}

This section introduces the classes of discrete and continuous switched systems that are of interest to this work, and provides the main theorems that establish boundedness of solutions under disturbances for sufficiently slow switching.

\subsection{Switched Discrete Systems} 
\label{subsec:system-description-discrete}

Let $\mathcal{P}$ be a finite index set and consider the family of discrete-time systems
\begin{align}\label{eq:system-family}
x_{k+1} = f_p(x_k, d_k), & & p\in\mathcal{P} \enspace,
\end{align}
where $x\in\mathbb{R}^n$ is the state of the system and $d_k \in \mathbb{R}^m$ is the value at time $k$ of the discrete disturbance signal $d$, which belongs to the set of \emph{bounded} disturbances $\mathcal{D}:=\{ d: \mathbb{Z}_+ \to \mathbb{R}^m ~|~\|d\|_\infty<\infty\}$.  It is assumed that, for each $p\in\mathcal{P}$, the mapping $f_p : \mathbb{R}^n\times\mathbb{R}^m \to \mathbb{R}^n$ is continuous in its arguments, and that there exists a unique $x_p^*\in\mathbb{R}^n$ which satisfies $x_p^* = f_p(x_p^*,0)$. Note that the vast majority of the relevant literature assumes that all subsystems $f_p$ share a common equilibrium point; here, we relax this assumption, and allow for $x_p^*\neq x_q^*$ when $p \neq q$. 

To state the main result, we will require each system in the family \eqref{eq:system-family} to be input-to-state stable, as defined below. 
\begin{definition}\label{def:ISS}
The system $f_p$ in \eqref{eq:system-family} is input-to-state stable (ISS) if there exists a class $\mathcal{KL}$ function $\beta$ and a class $\mathcal{K}_\infty$ function $\alpha$ such that for any initial state $x_0 \in \mathbb{R}^n$ and any bounded input $d \in \mathcal{D}$, the solution $x_k$ exists for all $k \geq 0$ and satisfies
\begin{equation}\label{eq:ISS_def}
\|x_k - x^*_p \| \leq \beta(\|x_0 - x^*_p \|, k) + \alpha(\|d\|_\infty) \enspace. 
\end{equation}
\end{definition}

Let $\sigma:\mathbb{Z}_+ \to \mathcal{P}$ be a switching signal, mapping the discrete time
$k$ to the index $\sigma(k)\in\mathcal{P}$ of the subsystem that is active at $k$. This gives rise to a discrete switched system of the form
\begin{align}\label{eq:switched-system}
x_{k+1} = f_{\sigma(k)}(x_k,d_k) \enspace.
\end{align}

We are interested in establishing boundedness and ultimate boundedness of the solutions of \eqref{eq:switched-system} under bounded disturbances, provided that the switching signal is sufficiently ``slow on average". Definition~\ref{def:average-dwell-time} below makes this notion precise.

\begin{definition}\label{def:average-dwell-time}
A switching signal $\sigma(k)$ has \emph{average dwell-time} $N_{\rm a}>0$ if the number $N_\sigma(k,\underline{k})\in\mathbb{Z}_+$ of switches over any discrete-time interval $[\underline{k},k)\cap \mathbb{Z}_+$ where $k,\underline{k}\in\mathbb{Z}_+$, satisfies
\begin{align} \label{eq:avg-dwell-time-def}
N_\sigma(k,\underline{k}) \leq N_0 + \frac{k-\underline{k}}{N_{\rm a}}\enspace, & & \forall k\geq \underline{k} \geq 0
\end{align}
where $N_0 >0$ is a finite constant.
\end{definition}

We can now state the main result of this section for discrete switched systems. 
\begin{theorem}\label{thm:ISS-switch-discrete}
Consider the switched system \eqref{eq:switched-system} and assume that for each $p \in \mathcal{P}$ there exists a continuous function $V_p:\mathbb{R}^n\to\mathbb{R}_+$ such that for all $x \in \mathbb{R}^n$ and $d\in\mathcal{D}$,
\begin{align}
\underline{\alpha}_p(\|x-x_p^*\|) \leq V_p(x) \leq \overline{\alpha}_p(\|x-x_p^*\|) \enspace, \label{eq:V-ISS-1}\\
V_p(x_{k+1}) \leq \lambda_p V_p(x_k) + \alpha_p(\|d\|_\infty) \label{eq:V-ISS-2-p}\enspace,
\end{align}
for any $k \geq 0$, where $0<\lambda_p<1$ and $\underline{\alpha}_p, \overline{\alpha}_p, \alpha_p$ are class $\mathcal{K}_\infty$ functions. Assume further that 
\begin{equation}\label{eq:Vratio_bound_discrete}
\limsup_{\|x-x^*_p\| \to \infty} \frac{V_q(x)}{V_p(x)} < \infty
\end{equation}
for any $p,q \in \mathcal{P}$. Then, there exists $\overline{N}_{\rm a} > 0$ so that for any switching signal $\sigma$ satisfying the average dwell-time constraint \eqref{eq:avg-dwell-time-def} with
\begin{equation} \label{eq:thm1_dwell_bound}
N_0 \geq 1 ~\mbox{and}~ N_{\rm a} \geq \overline{N}_{\rm a} \enspace, 
\end{equation}
and for any $x_0 \in \mathbb{R}^n$, there exists $K \in \mathbb{Z}_+$ such that the solution $\{x_k\}_{k \in \mathbb{Z}_+}$ of \eqref{eq:switched-system} satisfies:
\begin{enumerate}[(i)]
\item for all $0 \leq k < K$,
\begin{equation} \label{eq:transient_bound_discr}
\|x_k-x_{\sigma(k)}^*\| \leq \beta(\|x_0-x_{\sigma(0)}^*\|,k)+\alpha(\| d \|_\infty)
\end{equation}
for some $\beta\in\mathcal{KL}$ and $\alpha\in\mathcal{K}_\infty$;
\item for all $k \geq K$,
\begin{equation} \label{eq:ultimate_bound_discr}
x_k\in \mathcal{M}(\bar{\omega}) \!:=\!\! \bigcup_{p\in\mathcal{P}} \{x\in\mathbb{R}^n~|~V_p(x)\leq \bar{\omega}(\|d\|_\infty)\}
\end{equation}
where 
\begin{equation}\label{eq:omegabar-discrete}
\bar{\omega}(\|d\|_\infty):=c + \tilde{\alpha}(\|d\|_\infty) \enspace,
\end{equation}
for some $c>0$ and $\tilde{\alpha}\in\mathcal{K}_\infty$.
\end{enumerate} 
\end{theorem}
\noindent Proving Theorem~\ref{thm:ISS-switch-discrete} is postponed until Section~\ref{subsec:disc-proof}. Explicit expressions for the bound $\overline{N}_{\rm a}$ on the average dwell-time $N_{\rm a}$, and for $\bar{\omega}$ in the characterization of the set $\mathcal{M}(\bar{\omega})$ are important in applications, and are given before proofs; see Section~\ref{subsec:bounds-discrete}.

Let us now briefly discuss some aspects of Theorem~\ref{thm:ISS-switch-discrete}. First, note that by \cite[Definition~3.2]{jiang2001input}, conditions \eqref{eq:V-ISS-1}-\eqref{eq:V-ISS-2-p} imply that $V_p$ is an ISS-Lyapunov function for the $p$-th subsystem, which by \cite[Lemma~3.5]{jiang2001input}, entails that $f_p$ is ISS. 
Since we are interested in switching among systems from the family \eqref{eq:system-family}, condition \eqref{eq:Vratio_bound_discrete} is added to ensure that the ratio of the Lyapunov functions corresponding to the subsystems involved in switching is bounded.    
With conditions \eqref{eq:V-ISS-1}-\eqref{eq:V-ISS-2-p} and \eqref{eq:Vratio_bound_discrete} in place, Theorem~\ref{thm:ISS-switch-discrete} states that if each system in the family \eqref{eq:system-family} is ISS, the solution of \eqref{eq:switched-system} is uniformly\footnote{The term ``uniformly'' refers to uniformity over the set of switching signals that satisfy \eqref{eq:avg-dwell-time-def} for $N_0 \geq 1$ and $N_{\rm a} \geq \overline{N}_{\rm a}$, as required by \eqref{eq:thm1_dwell_bound} of Theorem~\ref{thm:ISS-switch-discrete}.} bounded and uniformly ultimately bounded within the compact set $\mathcal{M}(\bar{\omega})$ characterized by \eqref{eq:ultimate_bound_discr}. Furthermore, \eqref{eq:ultimate_bound_discr} indicates that the ``size" of $\mathcal{M}(\bar{\omega})$ reduces proportionally with the norm $\|d\|_\infty$ of the disturbance. Note, however, that $\mathcal{M}(\bar{\omega})$ does \emph{not} collapse to a point when the disturbance signal $d$ vanishes; indeed, if $d=0$, \eqref{eq:omegabar-discrete} implies that $\bar{\omega}(0)=c>0$ and the solutions of \eqref{eq:switched-system} are ultimately bounded to the zero-input compact set $\mathcal{M}(c)$ that contains the equilibria $x^*_p \in \mathcal{M}(c), ~\forall p \in \mathcal{P}$.

It should be noted that Theorem~\ref{thm:ISS-switch-discrete} does not establish ISS  for \eqref{eq:switched-system}, since the compact set $\mathcal{M}(c)$ is not positively invariant\footnote{In the terminology of \cite{sontag1996new}, $\mathcal{M}(\bar{\omega})$ is not a zero-invariant set for \eqref{eq:switched-system}.} under the zero-input dynamics of \eqref{eq:switched-system}. In fact, for suitable switching signals satisfying the requirements of the theorem, solutions of \eqref{eq:switched-system} can be found that start within $\mathcal{M}(c)$ and---while evolving in the absence of the disturbance---escape from $\mathcal{M}(c)$ before they return to $\mathcal{M}(c)$ and be trapped forever in it; see also Remark~\ref{rem:invariance} for how this behavior can emerge.
Note also that although the estimate \eqref{eq:transient_bound_discr} is reminiscent of \eqref{eq:ISS_def} in Definition~\ref{def:ISS}, it extends only up to a finite integer $K$, and it does not represent point-to-set distance from $\mathcal{M}(c)$ as establishing set-ISS for \eqref{eq:switched-system} would require. However, when all the subsystems in the family \eqref{eq:system-family} share the same equilibrium, then ISS can be recovered, as the following corollary shows. Corollary~\ref{cor:single-equ-disc} provides the counterpart of \cite[Theorem~3.1]{vu2007input} for discrete switched systems. 
\begin{corollary}\label{cor:single-equ-disc}
Consider \eqref{eq:switched-system} with $x_p^*=0$ for all $p\in\mathcal{P}$. Let the assumptions of Theorem~\ref{thm:ISS-switch-discrete} hold, and further assume that
\begin{equation}\label{eq:Vratio_lb_discrete}
\limsup_{\|x\| \to 0} \frac{V_q(x)}{V_p(x)} < \infty
\end{equation}
for all $p,q\in\mathcal{P}$. Then, the system \eqref{eq:switched-system} is ISS.
\end{corollary}

While the set $\mathcal{M}(\bar{\omega})$ in Theorem~\ref{thm:ISS-switch-discrete} is not positively invariant, one can identify a (compact) subset of initial conditions in $\mathcal{M}(c)$ such that the corresponding solutions never leave $\mathcal{M}(\bar{\omega})$. This property corresponds to the notion of \emph{practical stability} in the terminology of~\cite{kuiava2013practical}, and is important in certain motion planning applications; e.g., see~\cite{motahar2016composing, gregg-planning2012} for planning motions with limit-cycle walking bipedal robots. This is made precise by the following corollary.
\begin{corollary}\label{cor:trap-set}
Under the assumptions of Theorem~\ref{thm:ISS-switch-discrete}, there exists a non-empty compact set $\mathcal{S} \subset \mathcal{M}(c)$ such that $x_0 \in \mathcal{S}$, $d\in\mathcal{D}$ imply $x_k \in \mathcal{M}(\bar{\omega})$ for all $k\geq 0$.
\end{corollary}



\subsection{Switched Continuous Systems}

As in Section~\ref{subsec:system-description-discrete}, let $\mathcal{P}$ be a finite index set and consider the family of continuous-time systems
\begin{align}\label{eq:cont-family}
\dot{x}(t) = f_p(x(t),d(t)), & & p\in\mathcal{P}\enspace,
\end{align}
where $x\in\mathbb{R}^n$ is the state of the system and $d(t)\in\mathbb{R}^m$ is the value of the continuous-time disturbance signal $d$ at time $t$ which belongs to the set of \emph{bounded} disturbances $\mathcal{D}:=\{d:\mathbb{R}_+\to\mathbb{R}^m~|~\|d\|_\infty<\infty,~d~{\rm piecewise~continuous}\}$. It is assumed that, for each $p\in\mathcal{P}$, the vector field $f_p:\mathbb{R}^n\times\mathbb{R}^m \to \mathbb{R}^n$ is locally Lipschitz in its arguments, and that there exists a unique $x_p^* \in \mathbb{R}^n$ with $0=f_p(x_p^*,0)$. As in Section~\ref{subsec:system-description-discrete}, we allow for $x_p^*\neq x_q^*$ when $p \neq q$.

Analogous to Section~\ref{subsec:system-description-discrete}, we will require each system in the family \eqref{eq:cont-family} to be input-to-state stable, as defined below. 
\begin{definition}\label{def:ISS_cont}
The system $f_p$ in \eqref{eq:cont-family} is ISS if there exist a class $\mathcal{KL}$ function $\beta$ and a class $\mathcal{K}_\infty$ function $\alpha$ such that for any initial state $x(0) \in \mathbb{R}^n$ and any bounded input $d \in \mathcal{D}$ the solution $x(t)$ exists for all $t \geq 0$ and satisfies
\begin{equation}\label{eq:ISS_def_cont}
\|x(t) - x^*_p \| \leq \beta(\|x(0) - x^*_p \|, t) + \alpha(\|d\|_\infty) \enspace. 
\end{equation}
\end{definition}

Let $\sigma:\mathbb{R}_+\to\mathcal{P}$ be a switching signal mapping the time instant $t$ to the index $\sigma(t)\in\mathcal{P}$ of the subsystem that is active at $t$. It is assumed that $\sigma(t)$ is right-continuous. The switching signal gives rise to the continuous-time switched system
\begin{align}\label{eq:switched-system-cont}
\dot{x}(t) = f_{\sigma(t)}(x(t),d(t)) \enspace.
\end{align}
The solution $x(t) := \phi(t,x(0),\sigma(t),d(t))$ of \eqref{eq:switched-system-cont} is a sequential concatenation of each subsystem's solution as governed by the switching signal. Let $\{t_n\}_{n=1}^\infty$ with $t_n\in\mathbb{R}_+$ be a strictly monotonically increasing sequence of switching times. Clearly, continuity of $f_p(x,d)$ and piecewise continuity of $d(t)$ imply that $x(t)$ is continuous over $(t_n, t_{n+1})$, i.e. between subsequent switches. Furthermore, for any $t_n$, the subsystem $f_{\sigma(t_n)}$ that is switched in and is  active over $[t_n, t_{n+1})$ is initialized by $x(t_n) = \lim_{t\nearrow t_n} x(t)$ ensuring that $x(t)$ is continuous at $t_n$. Hence, $x(t)$ is continuous for all $t\geq 0$.

As in the discrete-time case, the main result of this section is stated for switching signals $\sigma$ with sufficiently slow switching on average; the following definition formalizes this notion. 

\begin{definition}\label{def:average-dwell-time-cont}
A switching signal $\sigma(t)$ has \emph{average dwell-time} $N_{\rm a}>0$ if the number $N_\sigma(t,\underline{t})\in\mathbb{Z}_+$ of switches over any interval $[\underline{t},t)\subset\mathbb{R}_+$ satisfies
\begin{align} \label{eq:avg-dwell-time-def-cont}
N_\sigma(t,\underline{t}) \leq N_0 + \frac{t-\underline{t}}{N_{\rm a}}\enspace, & & \forall t\geq \underline{t} \geq 0
\end{align}
where $N_0 >0$ is a finite constant.
\end{definition}

We are now in a position to state the main result of this section for continuous switched systems. 

\begin{theorem}\label{thm:ISS-switch-continuous}
Consider the switched system \eqref{eq:switched-system-cont} and assume that for each $p \in \mathcal{P}$ there exists a continuously differentiable function $V_p:\mathbb{R}^n\to\mathbb{R}_+$ such that for all $x \in \mathbb{R}^n$ and $d\in\mathcal{D}$,
\begin{align}
\underline{\alpha}_p(\|x-x_p^*\|) \leq V_p(x) \leq \overline{\alpha}_p(\|x-x_p^*\|) \enspace, \label{eq:V-ISS-p-1-cont} \\
\frac{\partial V_p}{\partial x} f_p(x,d) \leq -\lambda_p V_p(x) + \alpha_p(\|d\|_\infty) \enspace, \label{eq:V-ISS-p-2-cont}
\end{align}
for any $t \geq 0$, where $\lambda_p>0$ and $\underline{\alpha}_p, \overline{\alpha}_p, \alpha_p$ are class $\mathcal{K}_\infty$ functions. Assume further that 
\begin{equation}\label{eq:Vratio_bound_continuous}
\limsup_{\|x-x^*_p\| \to \infty} \frac{V_q(x)}{V_p(x)} < \infty
\end{equation}
for $p,q \in \mathcal{P}$. 
Then, there exists $\overline{N}_{\rm a} \!>\! 0$ so that for any switching signal $\sigma$ satisfying the average dwell-time constraint \eqref{eq:avg-dwell-time-def-cont} with
\begin{equation}\label{eq:thm2_dwell_bound}
 N_0 \geq 1 ~\mbox{and}~ N_{\rm a} \geq \overline{N}_{\rm a} \enspace, 
 \end{equation}
and for any $x(0) \in \mathbb{R}^n$, there exists $T \in \mathbb{R}_+$ such that the solution $x(t):=\phi(t,x(0),\sigma(t),d(t))$ of \eqref{eq:switched-system-cont} satisfies: 
\begin{enumerate}[(i)]
\item for all $0 \leq t < T$,
\begin{equation}\label{eq:transient_bound_cont}
\|x(t)-x_{\sigma(k)}^*\| \leq \beta(\|x(0)-x_{\sigma(0)}^*\|,t)+\alpha(\| d \|_\infty)
\end{equation}
for some $\beta\in\mathcal{KL}$, $\alpha\in\mathcal{K}_\infty$;
\item for all $t \geq T$,
\begin{equation} \label{eq:ultimate_bound_cont}
x(t) \!\in\! \mathcal{M}(\bar{\omega}) \!:=\!\! \bigcup_{p\in\mathcal{P}} \{x\in\mathbb{R}^n\!~|~\!V_p(x)\leq \bar{\omega}(\|d\|_\infty)\}
\end{equation}
where 
\begin{equation}\label{eq:omegabar-continuous}
\bar{\omega}(\|d\|_\infty):=c + \tilde{\alpha}(\|d\|_\infty) \enspace.
\end{equation}
for some $c>0$ and $\tilde{\alpha}\in\mathcal{K}_\infty$. 
\end{enumerate} 
\end{theorem}
\noindent A proof of Theorem~\ref{thm:ISS-switch-continuous} is presented in Section~\ref{subsec:cont_proof} below, and explicit expressions for the bound $\overline{N}_{\rm a}$ in \eqref{eq:thm2_dwell_bound} and for $\bar{\omega}$ in \eqref{eq:ultimate_bound_cont} are provided in Section~\ref{subsec:bounds-continuous}. It is only mentioned here that Theorem~\ref{thm:ISS-switch-continuous} is completely analogous to Theorem~\ref{thm:ISS-switch-discrete}, establishing uniform boundedness by \eqref{eq:transient_bound_cont} and uniform ultimate boundedness in the compact set $\mathcal{M}(\bar{\omega})$ characterized by \eqref{eq:ultimate_bound_cont} of the solutions of \eqref{eq:switched-system-cont}. Theorem~\ref{thm:ISS-switch-continuous} does not establish ISS stability of \eqref{eq:switched-system-cont} with respect to the compact set $\mathcal{M}(c)$, since this set is not invariant under the zero-input dynamics of \eqref{eq:switched-system-cont}; see Example 2 in Section~\ref{subsec:simulation} for an illustration of this behavior.
Finally, the following corollaries are the counterparts of Corollaries~\ref{cor:single-equ-disc} and~\ref{cor:trap-set} for the continuous switched system~\eqref{eq:switched-system-cont}. Note that Corollary~\ref{cor:single-equ-cont} particularizes Theorem~\ref{thm:ISS-switch-continuous} to \cite[Theorem~3.1]{vu2007input} for switched systems with a common equilibrium point.  

\begin{corollary}\label{cor:single-equ-cont}
Consider \eqref{eq:switched-system-cont} with $x_p^*=0$ for all $p\in\mathcal{P}$. Let the assumptions of Theorem~\ref{thm:ISS-switch-continuous} hold, and further assume that
\begin{equation}\label{eq:Vratio_lb_cont}
\limsup_{\|x\| \to 0} \frac{V_q(x)}{V_p(x)} < \infty
\end{equation}
for all $p,q\in\mathcal{P}$. Then, the system \eqref{eq:switched-system-cont} is ISS.
\end{corollary}

\begin{corollary}\label{cor:trap-set-cont}
Under the assumptions of Theorem~\ref{thm:ISS-switch-continuous}, there exists a non-empty compact set $\mathcal{S} \subset \mathcal{M}$ such that $x(0) \in \mathcal{S}$, $d\in\mathcal{D}$ imply $x(t) \in \mathcal{M}(\bar{\omega})$ for all $t\geq 0$.
\end{corollary}



\section{Set Constructions and Explicit Bounds}
\label{sec:discrete_switched_system}

This section characterizes the family of switching signals required by Theorems~\ref{thm:ISS-switch-discrete} and~\ref{thm:ISS-switch-continuous} by providing explicit expressions for the dwell-time bound $\overline{N}_{\rm a}$ in \eqref{eq:thm1_dwell_bound} and \eqref{eq:thm2_dwell_bound}, respectively. Explicit expressions of $\bar{\omega}$ are also given, thereby determining the sets $\mathcal{M}(\bar{\omega})$ within which the solutions ultimately converge. We begin with relevant set constructions motivated by \cite{basar2010}. 

\subsection{Set Constructions}
\label{subsec:set}

Suppose that $V_p$ is a function satisfying \eqref{eq:V-ISS-1}-\eqref{eq:V-ISS-2-p} in the case of discrete or \eqref{eq:V-ISS-p-1-cont}-\eqref{eq:V-ISS-p-2-cont} in the case of continuous switched systems. The $\kappa$-sublevel set of $V_p$ is defined as 
\[
\mathcal{M}_p(\kappa):=\{ x\in\mathbb{R}^n~|~V_p(x)\leq \kappa \}\enspace,
\]
and the union of the sublevel sets over $\mathcal{P}$ is denoted as
\begin{equation}\label{eq:M-union}
\mathcal{M}(\kappa) := \bigcup_{p\in\mathcal{P}} \mathcal{M}_p(\kappa) \enspace. 
\end{equation}
Next, we define a positive constant
\begin{equation}\label{eq:omega}
\omega(\kappa):=\max_{p\in\mathcal{P}} \max_{x\in \mathcal{M}(\kappa)} V_p(x)\enspace,
\end{equation}
which is well defined since $\mathcal{M}(\kappa)$ is compact for any $\kappa>0$ and $\mathcal{P}$ is finite. Intuitively, the definition of $\omega(\kappa)$ by \eqref{eq:omega} enlarges each sublevel set $\mathcal{M}_p(\kappa)$ so that the resulting enlarged set $\mathcal{M}_p(\omega(\kappa))$ includes the sets $\mathcal{M}_q(\kappa)$ for all $q \in \mathcal{P}$. An illustration of this construction can be seen in Fig.~\ref{fig:set-construction} and the following remark makes this intuition precise.
\begin{rem}\label{rem:kappa-intersect}
By the definition \eqref{eq:omega} of $\omega(\kappa)$, $V_p(x)\leq \omega(\kappa)$ for any $x\in\mathcal{M}(\kappa)$ and any $p\in\mathcal{P}$. Thus, $\mathcal{M}(\kappa)\subseteq\mathcal{M}_p(\omega(\kappa))$ for all $p\in\mathcal{P}$, implying that $\mathcal{M}(\kappa)\subseteq \bigcap_{p\in\mathcal{P}}\mathcal{M}_p(\omega(\kappa))$.
\end{rem}
 
\begin{figure}[t]
\vspace{+0.02in}
\begin{centering}
\includegraphics[width=0.7\columnwidth]{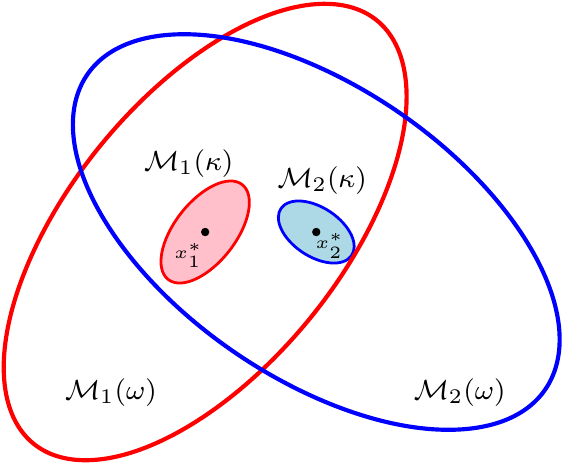} 
\par\end{centering}
\vspace{-0.1in}
\caption{Illustration of the set construction. The sublevel sets for system 1 are in red and the sublevel sets of system 2 are in blue; the construction leads to $\mathcal{M}_1(\kappa)\cup\mathcal{M}_2(\kappa) \subseteq \mathcal{M}_1(\omega)\cap\mathcal{M}_2(\omega)$ as Remark~\ref{rem:kappa-intersect} explains.}
\vspace{-0.15in}
\label{fig:set-construction} 
\end{figure}

We now establish a relationship between the values $V_p(x)$ and $V_q(x)$ of any pair $V_p, V_q$ of ISS-Lyapunov functions at a given point $x\in\mathbb{R}^n$ as the system switches between the corresponding two subsystems $p, q \in \mathcal{P}$, $p \neq q$. Consider the ratio $V_q(x) / V_p(x)$, and let
\begin{equation}\label{eq:mu-p-def}
\mu_p(\kappa):=\max_{q\in\mathcal{P}} \sup_{x \notin \accentset{\circ}{\mathcal{M}}_p(\kappa)} \frac{V_q(x)}{V_p(x)} \enspace, 
\end{equation}
%
which is bounded as shown in the following arguments. Let $a_1:=\limsup_{\|x-x_p^*\| \to \infty} V_q(x)/V_p(x)$ which is bounded by \eqref{eq:Vratio_bound_discrete} and \eqref{eq:Vratio_bound_continuous}. As a consequence of \cite[Theorem~3.17]{rudin1964principles} it follows that there exists a $r>0$ such that for any $x$ with $\|x-x_p^*\|>r$, we have $V_q(x)/V_p(x) \leq a_1 + 1$. Expand $r$ if necessary to ensure that $\mathcal{M}(\kappa) \subset B_r(x_p^*)$. Note that $V_q(x)/V_p(x)$ is continuous on $\mathbb{R}^n\setminus\{x_p^*\}$ hence it is also continuous on $\overline{B}_r(x_p^*)\setminus\accentset{\circ}{\mathcal{M}}_p(\kappa)\subset \mathbb{R}^n\setminus\{x_p^*\}$ which is compact. Then there exists $a_2>0$ such that $V_q(x)/V_p(x) < a_2$ for any $x\in \overline{B}_r(x_p^*)\setminus\accentset{\circ}{\mathcal{M}}_p(\kappa)$. Therefore $V_q(x)/V_p(x)<\max\{a_1+1,a_2\}$ for all $x\not\in\accentset{\circ}{\mathcal{M}}_p(\kappa)$ ensuring the boundedness of $\mu_p(\kappa)$ in \eqref{eq:mu-p-def}. 

This constant provides a bound on how much the value of the Lyapunov function can change on switching. Clearly, 
\begin{equation} \nonumber
\forall p, q \in \mathcal{P}, \enspace \enspace V_q(x) \leq \mu_p(\kappa) V_p(x) \enspace  \enspace \forall x \notin \accentset{\circ}{\mathcal{M}}_p(\kappa) \enspace.         
\end{equation}
To make this bound independent of $p$, let
\begin{equation} \label{eq:mu-def}
\mu(\kappa):=\max_{p\in\mathcal{P}}\mu_p(\kappa) \enspace,
\end{equation}
which implies that
\begin{equation} \label{eq:mu-bound-V}
\forall p, q \in \mathcal{P}, \enspace \enspace V_q(x) \leq \mu(\kappa) V_p(x) \enspace\enspace \forall x \notin \accentset{\circ}{\mathcal{M}}_p(\kappa) \enspace.
\end{equation}
Due to the interchangeability of the indices $p$ and $q$, it also holds that  $V_p(x) \leq \mu(\kappa) V_q(x)$ as long as $x \notin \accentset{\circ}{\mathcal{M}}_q(\kappa)$. Hence, when $x \notin \accentset{\circ}{\mathcal{M}}_p(\kappa) \cup \accentset{\circ}{\mathcal{M}}_q(\kappa)$, we can write $V_q(x) \leq \mu(\kappa)^2 V_q(x)$, from which it follows that 
\begin{equation}\label{eq:mu_1}
\mu(\kappa) \geq 1 \enspace,
\end{equation}
since $V_q$ is positive definite for $x \notin \accentset{\circ}{\mathcal{M}}_p(\kappa) \cup \accentset{\circ}{\mathcal{M}}_q(\kappa)$. 
Finally, in the context of the switched system \eqref{eq:switched-system}, it is worth noting that \eqref{eq:mu-bound-V} holds for a switching instant $k_n$ even if $x_{k_n} \in\accentset{\circ}{\mathcal{M}}(\kappa)$, as long as $x_{k_n} \notin \accentset{\circ}{\mathcal{M}}_{\sigma(k_n - 1)}(\kappa)$. A similar statement can be made for the  continuous switched system \eqref{eq:switched-system-cont}.

The aforementioned set constructions allow us to provide explicit expressions for the average dwell-time bound $\overline{N}_{\rm a}$ in \eqref{eq:thm1_dwell_bound} and \eqref{eq:thm2_dwell_bound}, and for the characterization of the sets $\mathcal{M}(\bar{\omega})$ in \eqref{eq:ultimate_bound_discr} and \eqref{eq:ultimate_bound_cont} of Theorems~\ref{thm:ISS-switch-discrete} and~\ref{thm:ISS-switch-continuous}, respectively. Although these expressions are derived in the proofs of Section~\ref{sec:proofs} below, we provide them  here to ease their use in applications.

\subsection{Switched Discrete Systems: Explicit Bounds}
\label{subsec:bounds-discrete}

For the sake of notational convenience in \eqref{eq:V-ISS-2-p}, let $\lambda:=\max_{p\in\mathcal{P}}\lambda_p$ and $\hat{\alpha}(\|d\|_\infty):=\max_{p\in\mathcal{P}} \{\alpha_p(\|d\|_\infty)\}$. Then, for any $p\in\mathcal{P}$ we can write \eqref{eq:V-ISS-2-p} as
\begin{align}\label{eq:V-ISS-2}
V_p(x_{k+1}) \leq \lambda V_p(x_k) + \hat{\alpha}(\|d\|_\infty) \enspace,
\end{align}
where $\lambda\in(0,1)$ and $\hat{\alpha}$ are independent of $p$.

Let $\delta$ be any constant that lies within $(\lambda,1)$, then it will be shown in the proof of Theorem~\ref{thm:ISS-switch-discrete} that the lower bound on the average dwell-time, i.e. $\overline{N}_{\rm a}$ in  \eqref{eq:thm1_dwell_bound}, is
\begin{equation}\label{eq:avg-dwell-lb-disc}
\overline{N}_{\rm a} = \frac{\ln\mu(\kappa)}{\ln (\delta/\lambda)} \enspace,
\end{equation}
where the constant $\mu(\kappa) \geq 1$ is defined by \eqref{eq:mu-def}. Furthermore, the compact set $\mathcal{M}(\bar{\omega})$ in Theorem~\ref{thm:ISS-switch-discrete}\emph{(ii)} is characterized by
\begin{align} \label{eq:omega-bar-disc}
\bar{\omega}:=\mu(\kappa)^{N_0}\omega(\kappa) + \frac{\mu(\kappa)^{N_0}}{1-\delta} \hat{\alpha}(\|d\|_\infty)   \enspace,
\end{align}
from which the constant $c>0$ and the class-$\mathcal{K}_\infty$ function $\tilde{\alpha}$ participating in \eqref{eq:omegabar-discrete} can be readily identified. 

It is remarked that the constants $\delta$ and $\kappa$ are design parameters available for tuning the frequency of the switching signal. The choice of $\delta\in(\lambda,1)$ provides a tradeoff between robustness of the switched system \eqref{eq:switched-system} and the switching frequency. In more detail, if $\delta$ is chosen close to $\lambda$, the lower bound on the average dwell-time \eqref{eq:avg-dwell-lb-disc} becomes large, thus limiting the number of  switches $N_\sigma(k,\underline{k})$ in any time interval $[\underline{k},k)$, as \eqref{eq:avg-dwell-time-def} implies. On the other hand, if $\delta$ is chosen close to $1$, the size of the compact set $\mathcal{M}({\bar{\omega}})$ to which the solutions ultimately converge increases as $\bar{\omega}$ in \eqref{eq:omega-bar-disc} increases. As a result, the solutions of \eqref{eq:switched-system} are permitted to wander in a larger set, indicating low robustness of \eqref{eq:switched-system} to disturbances. Hence, slower switching signals result in tighter trapping regions.

The effect of $\kappa$ is more involved. Picking smaller values of $\kappa$ results in inner sublevel level sets of $V_p$, thus reducing the size of $\mathcal{M}(\kappa)$ in \eqref{eq:M-union} and causing $\mu_p(\kappa)$ to increase as the supremum in \eqref{eq:mu-p-def} increases. Hence, smaller values of $\kappa$ result in larger values of $\mu(\kappa)$ by \eqref{eq:mu-def}, leading to slower switching frequencies as well as larger compact attractive sets $\mathcal{M}({\bar{\omega}})$; observe the role of $\mu(\kappa)$ in \eqref{eq:avg-dwell-lb-disc} and  \eqref{eq:omega-bar-disc}. On the other hand, picking a larger $\kappa$ will result in smaller $\mu(\kappa)$ allowing for faster switches, however, $\omega(\kappa)$ also increases making its effect on $\mathcal{\bar{\omega}}$ unclear; observe the $\mu^{N_0}\omega(\kappa)$ term in \eqref{eq:omega-bar-disc}.

\subsection{Switched Continuous Systems: Explicit Bounds}
\label{subsec:bounds-continuous}

To simplify notation in \eqref{eq:V-ISS-p-2-cont}, we define $\lambda:=\min_{p\in\mathcal{P}} \lambda_p>0$ and $\hat{\alpha}(\|d\|_\infty):=\max_{p\in\mathcal{P}} \{ \alpha_p(\|d\|_\infty) \}$. Then, for any $p\in\mathcal{P}$, 
\begin{align}\label{eq:V-ISS-cont-p-ind}
\frac{\partial V_p}{\partial x} f_p(x,d) \leq -\lambda V_p(x) + \hat{\alpha}(\|d\|_\infty) \enspace,
\end{align}
where $\lambda>0$ and $\hat{\alpha}(\|d\|_\infty)$ are independent of $p$. 

Let $\delta$ be any constant in the open interval $(0,\lambda)$, then the lower bound on the average dwell time $\overline{N}_{\rm a}$ in Theorem~\ref{thm:ISS-switch-continuous} is
\begin{equation}\label{eq:avg-dwell-lb-cont}
\overline{N}_{\rm a} = \frac{\ln\mu(\kappa)}{\lambda-\delta} \enspace,
\end{equation}
where $\mu(\kappa) \geq 1$ is defined by \eqref{eq:mu-def}.The compact set $\mathcal{M}(\bar{\omega})$ in Theorem~\ref{thm:ISS-switch-continuous}\emph{(ii)}, within which solutions of \eqref{eq:switched-system-cont} ultimately become trapped corresponds to
\begin{align} \label{eq:omega-bar-cont}
\bar{\omega}:=\mu(\kappa)^{1+N_0}\omega(\kappa) + \mu(\kappa)^{1+N_0}\frac{1}{\delta} \hat{\alpha}(\|d\|_\infty)   \enspace,
\end{align}
from which the constant $c>0$ and the class-$\mathcal{K}_\infty$ function $\tilde{\alpha}$ in \eqref{eq:omegabar-continuous} can be easily recognized.

As in the case of the discrete switched systems, the constant $\delta\in(0,\lambda)$ presents a tradeoff between the robustness of the system and the switching frequency; setting $\delta$ close to $0$ increases the disturbance term in \eqref{eq:omega-bar-cont} while setting $\delta$ close to $\lambda$ increases $\overline{N}_{\rm a}$ in \eqref{eq:avg-dwell-lb-cont}. Regarding the effect of $\kappa$, the discussion is identical to that in Section~\ref{subsec:bounds-discrete}.

\section{Proofs}
\label{sec:proofs}

This section presents proofs of Theorem~\ref{thm:ISS-switch-discrete} and Theorem~\ref{thm:ISS-switch-continuous} together with proofs of their corollaries.

\subsection{Switched Discrete Systems}
\label{subsec:disc-proof}

We begin by establishing an important estimate in the following lemma that will be used in the proof of Theorem~\ref{thm:ISS-switch-discrete}. 
\begin{lemma}\label{lem:V-sigma}
Consider \eqref{eq:switched-system}. Let $\underline{k} \in \mathbb{Z}_+$ be the initial time and $\{k_n\}_{n=1}^\infty$, $k_n\in\mathbb{Z}_+$, be a strictly monotonically increasing sequence of switching instants with $k_1>\underline{k}$. Given $\kappa> 0$, define $\mu(\kappa)$ by \eqref{eq:mu-p-def}-\eqref{eq:mu-def} and let 
\begin{equation}\label{eq:inf-to-kappa}
N:=\inf\{ n\in\mathbb{Z}_+\cup \{\infty\}~|~x_{k_n}\in\accentset{\circ}{\mathcal{M}}_{\sigma(k_{n}-1)}(\kappa)\} \enspace ,
\end{equation}
be the index of the first switching instant $k_N$ for which \eqref{eq:mu-bound-V} cannot be used. Assume that for each $p \in \mathcal{P}$, the function $V_p$ satisfies \eqref{eq:V-ISS-1} and \eqref{eq:V-ISS-2} for some $\lambda \in (0,1)$. Choose $\delta\in(\lambda,1)$ and assume that the switching signal satisfies Definition~\ref{def:average-dwell-time} for any $N_0\geq 1$ and $N_{\rm a}\geq \overline{N}_{\rm a}$ where $\overline{N}_{\rm a}$ is given by \eqref{eq:avg-dwell-lb-disc}. Then, for any $x_{\underline{k}}\in\mathbb{R}^n$, $d\in\mathcal{D}$, and $\underline{k}\leq k <k_N$,
\begin{align}\label{eq:ISS}
V_{\sigma(k)}(x_k) \leq \mu^{N_0} \delta^{k-\underline{k}} V_{\sigma(\underline{k})}(x_{\underline{k}}) + \frac{\mu^{N_0}}{1-\delta} \hat{\alpha}(\|d\|_\infty) \enspace,
\end{align}
where $\hat{\alpha}$ is the class-$\mathcal{K}_\infty$ function in \eqref{eq:V-ISS-2}.
In addition, for $\underline{k}\leq k <k_N$ the solutions of \eqref{eq:switched-system} satisfy
\begin{equation}\label{eq:beta-alpha-bound-disc}
\|x_k-x_{\sigma(k)}^*\| \leq \beta(\|x_{\underline{k}}-x_{\sigma(\underline{k})}^*\|,k - \underline{k}) + \alpha(\|d\|_\infty)
\end{equation}
for some $\beta\in\mathcal{KL}$ and $\alpha\in\mathcal{K}_\infty$. 
\end{lemma}

The proof of the lemma can be found in Appendix~\ref{app:lemma-disc}. Now we are ready to present the proof of Theorem~\ref{thm:ISS-switch-discrete}.

\begin{proof}[Proof of Theorem~\ref{thm:ISS-switch-discrete}]
The arguments of the proof  refer to the set constructions of Section~\ref{subsec:set}. To simplify notation, the dependence on $\kappa$ of $\omega(\kappa)$ in \eqref{eq:omega} and of $\mu(\kappa)$ in \eqref{eq:mu-def} will be dropped. Consider any (fixed) switching signal $\sigma: \mathbb{Z}_+ \to \mathcal{P}$ satisfying Definition~\ref{def:average-dwell-time} for any $N_0\geq 1$ and $N_{\rm a}\geq \overline{N}_{\rm a}$ where $\overline{N}_{\rm a}$ is given by \eqref{eq:avg-dwell-lb-disc}. Without loss of generality, assume that the system starts at $k=0$ and let $\{k_1, k_2,...\}$ be the switching times. We first prove part (ii) and then part (i) of the theorem. 

For part (ii), we distinguish the following cases:

\noindent \emph{Case (a):} $V_{\sigma(0)}(x_0) \leq \omega$.\\ 
If $N$ in \eqref{eq:inf-to-kappa} is unbounded, \eqref{eq:mu-bound-V} can be used at all switching times and Lemma~\ref{lem:V-sigma} ensures that \eqref{eq:ISS} holds for all $k\geq \underline{k} = 0$. Since $V_{\sigma(0)}(x_0) \leq \omega$ and $\delta \in (\lambda,1)$, \eqref{eq:ISS} implies
\begin{align}
V_{\sigma(k)}(x_k) \leq \mu^{N_0}\omega + \frac{\mu^{N_0}}{1-\delta} \hat{\alpha}(\|d\|_\infty)=:\bar{\omega} \enspace \label{eq:N-inf-omega}
\end{align}
for all $k\geq 0$, showing that $x_k\in\mathcal{M}(\bar{\omega})$ for all $k\geq 0$ with $\bar{\omega}$ as in \eqref{eq:N-inf-omega}.
When, on the other hand, $N$ is a finite number in $\mathbb{Z}_+$, Lemma~\ref{lem:V-sigma} ensures that the estimate \eqref{eq:N-inf-omega} holds over the interval $[0, k_N)$. 
By \eqref{eq:inf-to-kappa} it is clear that $x_{k_N}\in\accentset{\circ}{\mathcal{M}}_{\sigma(k_{N}-1)}(\kappa)\subset \mathcal{M}(\kappa)$, which by Remark~\ref{rem:kappa-intersect} implies 
\begin{equation}\label{eq:bound_k_N}
V_{\sigma(k_N)}(x_{k_N})\leq \omega \enspace.
\end{equation}
Since $\mu \geq 1$, the definition of $\bar{\omega}$ by \eqref{eq:N-inf-omega} implies that $\omega \leq \bar{\omega}$, so that by \eqref{eq:bound_k_N} we have $V_{\sigma(k_N)}(x_{k_N})\leq \bar{\omega}$. As a result, when $N$ is finite, the validity of the estimate \eqref{eq:N-inf-omega} can be extended over the interval $[0, k_N]$. 
Now, considering $k=k_N$ as the initial instant, the initial condition $x_{k_N}$ satisfies \eqref{eq:bound_k_N} and the requirement for Case (a) holds at $k=k_N$. Hence, applying \eqref{eq:ISS} of Lemma~\ref{lem:V-sigma} with $\underline{k}=k_N$ and propagating the same arguments as above from $k_N$ onwards shows that $V_{\sigma(k)}(x_k) \leq \bar{\omega}$ for all $k\geq 0$, proving that $x_k$ never escapes from $\mathcal{M}(\bar{\omega})$. By the expression \eqref{eq:N-inf-omega} for $\bar{\omega}$, choosing $c:=\mu^{N_0}\omega$ and $\tilde{\alpha}(\|d\|_\infty):= (\mu^{N_0}/(1-\delta)) \hat{\alpha}(\|d\|_\infty)$ in \eqref{eq:omegabar-discrete} proves part (ii) for Case (a).  

\noindent \emph{Case (b)} $V_{\sigma(0)}(x_0) > \omega$.\\ 
As in Case (a), we distinguish between two subcases based on $N$ defined by \eqref{eq:inf-to-kappa}. When $N$ is unbounded, Lemma~\ref{lem:V-sigma} ensures that \eqref{eq:ISS} holds for all $k\geq \underline{k} = 0$; that is,
\begin{align}\label{eq:ISS_k0}
V_{\sigma(k)}(x_k) \leq \mu^{N_0} \delta^k V_{\sigma(0)}(x_0) + \frac{\mu^{N_0}}{1-\delta} \hat{\alpha}(\|d\|_\infty) \enspace.
\end{align}
If $K\in\mathbb{Z}_+$ is such that
\begin{equation}\label{eq:K_Nunbounded}
K\geq \frac{\ln{(V_{\sigma(0)}(x_0)/\omega)} }{\ln{(1/\delta})} \enspace,
\end{equation}
then $\delta^k V_{\sigma(0)}(x_0) \leq \omega$ for all $k\geq K$, and \eqref{eq:ISS_k0} implies that the bound \eqref{eq:N-inf-omega} holds
for all $k\geq K$, establishing that $x_k\in\mathcal{M}(\bar{\omega})$ for all $k \geq K$. If, on the other hand, $N$ is a finite integer in $\mathbb{Z}_+$, then by the definition of $N$ in \eqref{eq:inf-to-kappa} we have $x_{k_N}\in\accentset{\circ}{\mathcal{M}}_{\sigma(k_{N}-1)}(\kappa)\subset \mathcal{M}(\kappa)$. By Remark~\ref{rem:kappa-intersect}, this condition implies that $V_{\sigma(k_N)}(x_{k_N})\leq \omega$ and the state $x_{k_N}$ satisfies the conditions for Case (a). Hence, repeating the arguments of Case (a) from $k_N$ onwards with $x_{k_N}$ as the initial condition shows  that $x_k\in\mathcal{M}(\bar{\omega})$ for all $k\geq k_N$ and with $\bar{\omega}$ as defined in \eqref{eq:N-inf-omega}. Thus, choosing $K = k_N$ proves part (ii) for Case (b) with the same choice for $c$ and $\tilde{\alpha}$ in \eqref{eq:omegabar-discrete} as in Case (a).

For part (i), when the initial condition $x_0$ satisfies Case (a), then $K=0$ and the statement is vacuously true. If, on the other hand, $x_0$ satisfies the conditions of Case (b), observe from the arguments above that $K < k_N + 1$. Indeed, if $N$ is unbounded, $k_N \to \infty$ and $K$ is given by \eqref{eq:K_Nunbounded} while if $N$ is a finite integer, $K$ was selected equal to $k_N$. Hence, \eqref{eq:beta-alpha-bound-disc} in Lemma~\ref{lem:V-sigma} holds for all $k$ with $\underline{k}=0\leq k < K$, and the proof of part (i) is completed by choosing $\beta$, $\alpha$ as in Lemma~\ref{lem:V-sigma}.
\end{proof}

\begin{rem}\label{rem:invariance}
It is of interest to discuss the behavior of the set $\mathcal{M}(\bar{\omega})$ in the absence of  disturbances; that is, when $d_k=0$ for all $k \in \mathbb{Z}_+$. In this case, $\bar{\omega} = c = \mu^{N_0} \omega \geq \omega$. It is clear from the proof of Theorem~\ref{thm:ISS-switch-discrete} that if the initial conditions $x_0$ satisfy $V_{\sigma(0)}(x_0) \leq \omega \leq c$, the solution never leaves the set $\mathcal{M}(c)$. However, this does not imply that $\mathcal{M}(c)$ is a forward invariant set of the 0-input system. Indeed, if the initial conditions $x_0$ belong in the set $\mathcal{M}(c)$ but satisfy $\omega < V_{\sigma(0)}(x_0) \leq c$, the solution may exit $\mathcal{M}(c)$ before it returns to it forever, as the proof of Case (b) indicates. Example 2 in Section~\ref{subsec:simulation} illustrates that such behavior is possible.
\end{rem}

\begin{proof}[Proof of Corollary~\ref{cor:single-equ-disc}]
With the additional assumption\footnote{Essentially, \eqref{eq:Vratio_lb_discrete} ensures that $V_p(x)$ does not converge to 0 substantially faster than $V_q(x)$ as $x\to 0$.} \eqref{eq:Vratio_lb_discrete}, \eqref{eq:mu-p-def} is bounded over the entire $\mathbb{R}^n$ without the exclusion of an open set containing $0$. Thus, $\mu$ can be used for switches occurring at any $x \in \mathbb{R}^n$ and hence $k_N\to\infty$ in Lemma~\ref{lem:V-sigma} so that \eqref{eq:beta-alpha-bound-disc} holds for all $k \geq 0$. 
\end{proof}

\begin{proof}[Proof of Corollary~\ref{cor:trap-set}]
The proof is immediate by noting that $K=0$ for any $x_0\in\cap_{p\in\mathcal{P}}\mathcal{M}_p(\omega)=:\mathcal{S}$ due to the fact that Case (a) in the proof of Theorem~\ref{thm:ISS-switch-discrete} holds for this set. Further Remark~\ref{rem:kappa-intersect} ensures that $\mathcal{S}\neq\emptyset$.
\end{proof}

\subsection{Switched Continuous Systems}
\label{subsec:cont_proof}

We begin with the following lemma, which is analogous to Lemma~\ref{lem:V-sigma} and will be used in the proof of Theorem~\ref{thm:ISS-switch-continuous}.

\begin{lemma}\label{lem:V-sigma-cont}
Consider \eqref{eq:switched-system-cont}. Let $\underline{t} \in \mathbb{R}_+$ be the initial time and $\{t_n\}_{n=1}^\infty$, $t_n\in\mathbb{R}_+$, be a strictly monotonically increasing sequence of switching instants with $t_1>\underline{t}$. Let $x(t)$ be the solution of \eqref{eq:switched-system-cont} for the corresponding switching signal. Given $\kappa>0$, define $\mu(\kappa)$ by \eqref{eq:mu-p-def}-\eqref{eq:mu-def}  and let\footnote{Define $\sigma(t_n^-):=\lim_{t\nearrow t_n}\sigma(t)$. }
\begin{equation}\label{eq:inf-to-kappa-cont}
N:=\inf\{ n\in\mathbb{Z}_+\cup \{\infty\}~|~x(t_n)\in\accentset{\circ}{\mathcal{M}}_{\sigma(t_n^-)}(\kappa)\} \enspace,
\end{equation}
be the index of the first switching instant $t_N$ for which \eqref{eq:mu-bound-V} cannot be used. Assume that for each $p \in \mathcal{P}$, the function $V_p$ satisfies \eqref{eq:V-ISS-p-1-cont} and \eqref{eq:V-ISS-cont-p-ind} for some $\lambda>0$. Choose $\delta\in(0,\lambda)$ and assume that the switching signal satisfies Definition~\ref{def:average-dwell-time-cont} for any $N_0\geq 1$ and $N_{\rm a}\geq \overline{N}_{\rm a}$ where $\overline{N}_{\rm a}$ is given by \eqref{eq:avg-dwell-lb-cont}. Then, for any $x(\underline{t})\in\mathbb{R}^n$, $d\in\mathcal{D}$, and $\underline{t} \leq t <t_N$,
\begin{equation}\label{eq:ISS-cont}
V_{\sigma(t)}(x(t)) \leq \mu^{1+N_0} \mathrm{e}^{-\delta (t-\underline{t})} V_{\sigma(\underline{t})}(x(\underline{t})) + \frac{\mu^{1+N_0}}{\delta} \hat{\alpha}(\|d\|_\infty) \enspace,
\end{equation}
where $\hat{\alpha}$ is the class-$\mathcal{K}_\infty$ function in \eqref{eq:V-ISS-cont-p-ind}. In addition, for $\underline{t} \leq t < t_N$ the solutions of \eqref{eq:switched-system-cont} satisfy
\begin{equation}\label{eq:beta-alpha-bound-cont}
\|x(t)-x_{\sigma(t)}^*\| \leq \beta(\|x(\underline{t})-x_{\sigma(\underline{t})}^*\|,t- \underline{t}) + \alpha(\|d\|_\infty)
\end{equation}
for some $\beta\in\mathcal{KL}$ and $\alpha\in\mathcal{K}_\infty$. 
\end{lemma}
\noindent The proof for Lemma~\ref{lem:V-sigma-cont} can be found in Appendix~\ref{app:lemma-cont}. Now we present the proof of Theorem~\ref{thm:ISS-switch-continuous}.

\begin{proof}[Proof of Theorem~\ref{thm:ISS-switch-continuous}]
The proof is similar to that of Theorem~\ref{thm:ISS-switch-discrete}. With reference to the set constructions of Section~\ref{subsec:set} we define $\omega$ as in \eqref{eq:omega} and $\mu$ as in \eqref{eq:mu-def}. Consider any (fixed) switching signal $\sigma: \mathbb{R}_+ \to \mathcal{P}$ that satisfies Definition~\ref{def:average-dwell-time-cont} for any $N_0\geq 1$ and $N_{\rm a}\geq \overline{N}_{\rm a}$ where $\overline{N}_{\rm a}$ is given by \eqref{eq:avg-dwell-lb-cont}. Without loss of generality, assume that the system starts at $t=0$ and let $\{t_1, t_2,...\}$ be the corresponding switching times. We begin by proving part (ii), followed by part (i). 

For part (ii), we distinguish the following cases:

\noindent \emph{Case (a):} $V_{\sigma(0)}(x(0)) \leq \omega$.\\
First, consider the case where $N$ in \eqref{eq:inf-to-kappa-cont} is unbounded. Then, due to the choice of the switching signal, Lemma~\ref{lem:V-sigma-cont} holds and ensures that \eqref{eq:ISS-cont} holds for all $t\geq \underline{t} = 0$. Since $V_{\sigma(0)}(x(0)) \leq \omega$ and $\delta > 0$, \eqref{eq:ISS-cont} implies
\begin{align}
V_{\sigma(t)}(x(t)) \leq \mu^{1+N_0} \omega +\mu^{1+N_0} \frac{1}{\delta} \hat{\alpha}(\|d\|_\infty) =: \bar{\omega} \enspace, \label{eq:N-inf-omega-cont}
\end{align}
for all $t\geq 0$, showing that $x(t)\in\mathcal{M}(\bar{\omega})$ for all $t\geq 0$ with $\bar{\omega}$ as in \eqref{eq:N-inf-omega-cont}.
Now, consider the case where $N$ is a finite index in $\mathbb{Z}_+$. Then, Lemma~\ref{lem:V-sigma-cont} ensures that the bound \eqref{eq:N-inf-omega-cont} holds over the interval $[0, t_N)$. By the definition of $N$ in \eqref{eq:inf-to-kappa-cont} it follows that $x(t_N)\in\accentset{\circ}{\mathcal{M}}_{\sigma(t_N^-)}(\kappa)\subset \mathcal{M}(\kappa)$; thus, Remark~\ref{rem:kappa-intersect} implies
\begin{equation}\label{eq:bound_t_N}
V_{\sigma(t_N)}(x(t_N))\leq \omega \enspace.
\end{equation}
Since $\mu\geq 1$, the definition of $\bar{\omega}$ by \eqref{eq:N-inf-omega-cont} implies that $\omega \leq \bar{\omega}$ and so $V_{\sigma(t_N)}(x(t_N))\leq \bar{\omega}$. Hence, the bound \eqref{eq:N-inf-omega-cont} holds over the closed interval $[0, t_N]$. Furthermore, since starting at time $t_N$ with initial condition $x(t_N)$ satisfies \eqref{eq:bound_t_N}, Case (a) continues to hold, and the same arguments as above can be repeated from $t_N$ onwards. As a result,  $V_{\sigma(t_N)}(x(t_N))\leq \bar{\omega}$ for all $t \geq 0$, thereby proving that the solution $x(t)$ of \eqref{eq:switched-system-cont} is trapped in $\mathcal{M}(\bar{\omega})$. By the expression \eqref{eq:N-inf-omega-cont} for the bound $\bar{\omega}$, choosing $c:=\mu^{1+N_0}\omega$ and $\tilde{\alpha}(\|d\|_\infty):= (\mu^{1+N_0}/\delta) \hat{\alpha}(\|d\|_\infty)$ in \eqref{eq:omegabar-continuous} proves part (ii) for Case (a).  

\noindent \emph{Case (b):} $V_{\sigma(0)}(x(0)) > \omega$.\\
Again, we differentiate two cases based on $N$. If $N$ is unbounded, then \eqref{eq:ISS-cont} holds for all $t\geq \underline{t} = 0$; that is,
\begin{equation}\label{eq:ISS-cont-t0}
V_{\sigma(t)}(x(t)) \leq \mu^{1+N_0} \mathrm{e}^{-\delta t} V_{\sigma(0)}(x(0)) + \frac{\mu^{1+N_0}}{\delta} \hat{\alpha}(\|d\|_\infty) \enspace.
\end{equation}
Clearly, if $T\in\mathbb{R}_+$ satisfies
\[
T\geq \frac{\ln{(V_{\sigma(0)}(x(0))/\omega)} }{\delta} \enspace,
\]
then $\mathrm{e}^{-\delta t} V_{\sigma(0)}(x(0)) \leq \omega$ for all $t\geq T$, and \eqref{eq:ISS-cont-t0} implies that the estimate \eqref{eq:N-inf-omega-cont} holds for all $t \geq T$. As a result, $x(t) \in \mathcal{M}(\bar{\omega})$ for all $t \geq T$. If $N$ is finite, then by the definition \eqref{eq:inf-to-kappa-cont} of $N$, we have $x(t_N)\in\accentset{\circ}{\mathcal{M}}_{\sigma(t_N^-)}(\kappa)\subset \mathcal{M}(\kappa)$. By Remark~\ref{rem:kappa-intersect}, this condition implies that $V_{\sigma(t_N)}(x(t_N))\leq \omega$ and the same arguments as in Case (a) can be applied from $t_N$ onwards. Hence, selecting $T = t_N$ proves part (ii) for Case (b) with $c$ and $\tilde{\alpha}$ as elected in Case (a) above.

For part (i), when the initial condition $x(0)$ satisfies Case (a), then $T=0$ and the statement is trivially true. If, on the other hand, $x(0)$ satisfies Case (b), then $T \leq t_N$ and \eqref{eq:beta-alpha-bound-cont} in Lemma~\ref{lem:V-sigma-cont} holds for all $0 \leq t < T$. Thus, the proof of part (i) is completed by choosing the functions $\beta$, $\alpha$ as in Lemma~\ref{lem:V-sigma-cont}.
%
\end{proof}

\begin{proof}[Proof of Corollary~\ref{cor:single-equ-cont}]
With the additional assumption \eqref{eq:Vratio_lb_cont}, \eqref{eq:mu-p-def} holds over the entire $\mathbb{R}^n$ without the exclusion of an open set containing $0$. 
Thus, $\mu$ can be used for switches occurring at any $x \in \mathbb{R}^n$ and hence $t_N\to\infty$ in Lemma~\ref{lem:V-sigma-cont} so that \eqref{eq:beta-alpha-bound-cont} holds for all $t \geq 0$. 
\end{proof}

\begin{proof}[Proof of Corollary~\ref{cor:trap-set-cont}]
The proof is immediate by noting that $T=0$ for any $x(0)\in\cap_{p\in\mathcal{P}}\mathcal{M}_p(\omega)=:\mathcal{S}$ because Case (a) in the proof of Theorem~\ref{thm:ISS-switch-continuous} holds for this set. Further Remark~\ref{rem:kappa-intersect} ensures that $\mathcal{S}\neq\emptyset$.
\end{proof}

\section{Implementation Aspects}
\label{sec:examples}

This section addresses certain aspects that are of practical interest regarding the 
application of Theorems~\ref{thm:ISS-switch-discrete} and~\ref{thm:ISS-switch-continuous}.

\subsection{Analytical Bounds For Quadratic Lyapunov Functions}

Given the design parameters $\kappa$ and $\delta$, computing the average dwell time bound $\overline{N}_{\rm a}$ by \eqref{eq:avg-dwell-lb-disc} or \eqref{eq:avg-dwell-lb-cont}, and the estimate $\mathcal{M}({\bar{\omega}})$ of the compact set within which solutions are ultimately bounded by \eqref{eq:omega-bar-disc} or \eqref{eq:omega-bar-cont} requires the computation of $\mu(\kappa)$ and $\omega(\kappa)$, which can be challenging. This fact has been pointed in~\cite{basar2010}, where the authors highlighted the need for efficient tools for computing $\mu(\kappa)$ and $\omega(\kappa)$ as numerical computations based on discretizing the state-space become impractical as the dimension of the system grows. However, in the case where quadratic functions $V_p$ are used, analytical bounds for $\mu(\kappa)$ and $\omega(\kappa)$ can be found, as the following proposition shows.

\begin{prop}\label{prop:mu-compute}
Let $V_p(x)=(x-x_p^*)^{\rm T} S_p (x-x_p^*)$ for all $p\in\mathcal{P}$ be a family of positive definite quadratic functions and $\lambda_{\min}(S_p)$ be the minimum and $\lambda_{\max}(S_p)$  be the maximum eigenvalues of $S_p$, respectively. Given $\kappa>0$, define $\omega(\kappa)$ by \eqref{eq:omega} and $\mu(\kappa)$ by \eqref{eq:mu-def}. Then, the following hold

\vspace{-5mm}
\small 
\begin{align}
\omega(\kappa) & \leq \max_{p,q\in\mathcal{P}} \Bigg( \lambda_{\max}(S_p) \bigg( \sqrt{\frac{\kappa}{\lambda_{\min}(S_q)}} + \|x_p^*-x_q^*\| \bigg)^2 \Bigg) \label{eq:omega-bound} \\
\mu(\kappa) & \leq \max_{p,q\in\mathcal{P}} \Bigg( \frac{\lambda_{\max}(S_q)}{\lambda_{\min}(S_p)}\Big( 1 + \sqrt{\frac{\lambda_{\max}(S_p)}{\kappa}}\|x_p^*-x_q^*\| \Big)^2 \Bigg). \label{eq:mu-bound}
\end{align}
\normalsize
\end{prop}
\begin{proof}
Since the functions $V_p$ are quadratic, for all $x\in\mathbb{R}^n$
\begin{equation}\label{eq:quad-V-ISS-1}
\lambda_{\min}(S_p)\|x-x_p^*\|^2 \leq V_p(x) \leq \lambda_{\max}(S_p)\|x-x_p^*\|^2 \enspace. 
\end{equation}

We will first show \eqref{eq:omega-bound}. From \eqref{eq:M-union}, \eqref{eq:omega} and since $\mathcal{P}$ is a finite set, it follows that 
\begin{equation}\label{eq:omega-transform}
\omega(\kappa):=\max_{p\in\mathcal{P}} \max_{x\in \mathcal{M}(\kappa)} V_p(x) = \max_{p,q\in\mathcal{P}} \max_{x\in\mathcal{M}_q(\kappa)}V_p(x).
\end{equation}
Consider $\max_{x\in\mathcal{M}_q(\kappa)} V_p(x)$. For any $x\in\mathcal{M}_q(\kappa)$, we have
\begin{align}
V_p(x) & \leq \lambda_{\max}(S_p) \|x-x_p^*\|^2 \label{eq:omega-1}\\
& \leq \lambda_{\max}(S_p) \big( \|x-x_q^*\| + \|x_q^* - x_p^*\| \big)^2 \label{eq:omega-2}\\
& \leq \lambda_{\max}(S_p) \bigg( \sqrt{\frac{\kappa}{\lambda_{\min}(S_q)}} + \|x_q^* - x_p^*\| \bigg)^2 \enspace, \label{eq:omega-3}
\end{align}
where \eqref{eq:omega-1} follows from the second inequality of \eqref{eq:quad-V-ISS-1}, which further leads to \eqref{eq:omega-2} by the use of triangle inequality. 
Finally, \eqref{eq:omega-3} follows from noting that for any $x\in\mathcal{M}_q(\kappa)$, the first inequality of \eqref{eq:quad-V-ISS-1} provides the bound $\|x-x_q^*\|\leq \sqrt{\kappa/\lambda_{\min}(S_q)}$, which, on using in \eqref{eq:omega-2}, gives \eqref{eq:omega-3}. As \eqref{eq:omega-3} holds for any $x\in\mathcal{M}_q(\kappa)$, we have shown that $\max_{x\in\mathcal{M}_q(\kappa)} V_p(x)$ satisfies the bound in \eqref{eq:omega-3}, which by \eqref{eq:omega-transform} gives \eqref{eq:omega-bound}.

To show \eqref{eq:mu-bound}, from \eqref{eq:mu-p-def}, \eqref{eq:mu-def} and the finite $\mathcal{P}$ we have 
\begin{equation}\label{eq:mu-transform}
\mu(\kappa) = \max_{p,q\in\mathcal{P}} \sup_{x\not\in\accentset{\circ}{\mathcal{M}}_p(\kappa)} \frac{V_q(x)}{V_p(x)}.
\end{equation}
Consider $\sup_{x\not\in\accentset{\circ}{\mathcal{M}}_p(\kappa)} V_q(x)/V_p(x)$. For any $x\not\in\accentset{\circ}{\mathcal{M}}_p(\kappa)$,
\begin{align}
\frac{V_q(x)}{V_p(x)} & \leq \frac{\lambda_{\max}(S_q)\|x-x_q^*\|^2}{\lambda_{\min}(S_p)\|x-x_p^*\|^2} \label{eq:mu-1}\\
& \leq \frac{\lambda_{\max}(S_q)}{\lambda_{\min}(S_p)}\Big( 1 + \frac{\|x_p^*-x_q^*\|}{\|x-x_p^*\|} \Big)^2 \enspace, \label{eq:mu-2}
\end{align}
where \eqref{eq:mu-1} follows from \eqref{eq:quad-V-ISS-1}, and \eqref{eq:mu-2} follows from the triangle inequality. For $x\not\in\accentset{\circ}{\mathcal{M}}_p(\kappa)$, $V_p(x)\geq \kappa$ which by the second inequality of \eqref{eq:quad-V-ISS-1} gives $\|x-x_p^*\|\geq \sqrt{\kappa/\lambda_{\max}(S_p)}$. Using this in \eqref{eq:mu-2} followed by \eqref{eq:mu-transform} gives \eqref{eq:mu-bound}.
\end{proof}

\begin{figure*}[t]
\centering
\subfigure[]
{
\includegraphics[width=0.4\textwidth]{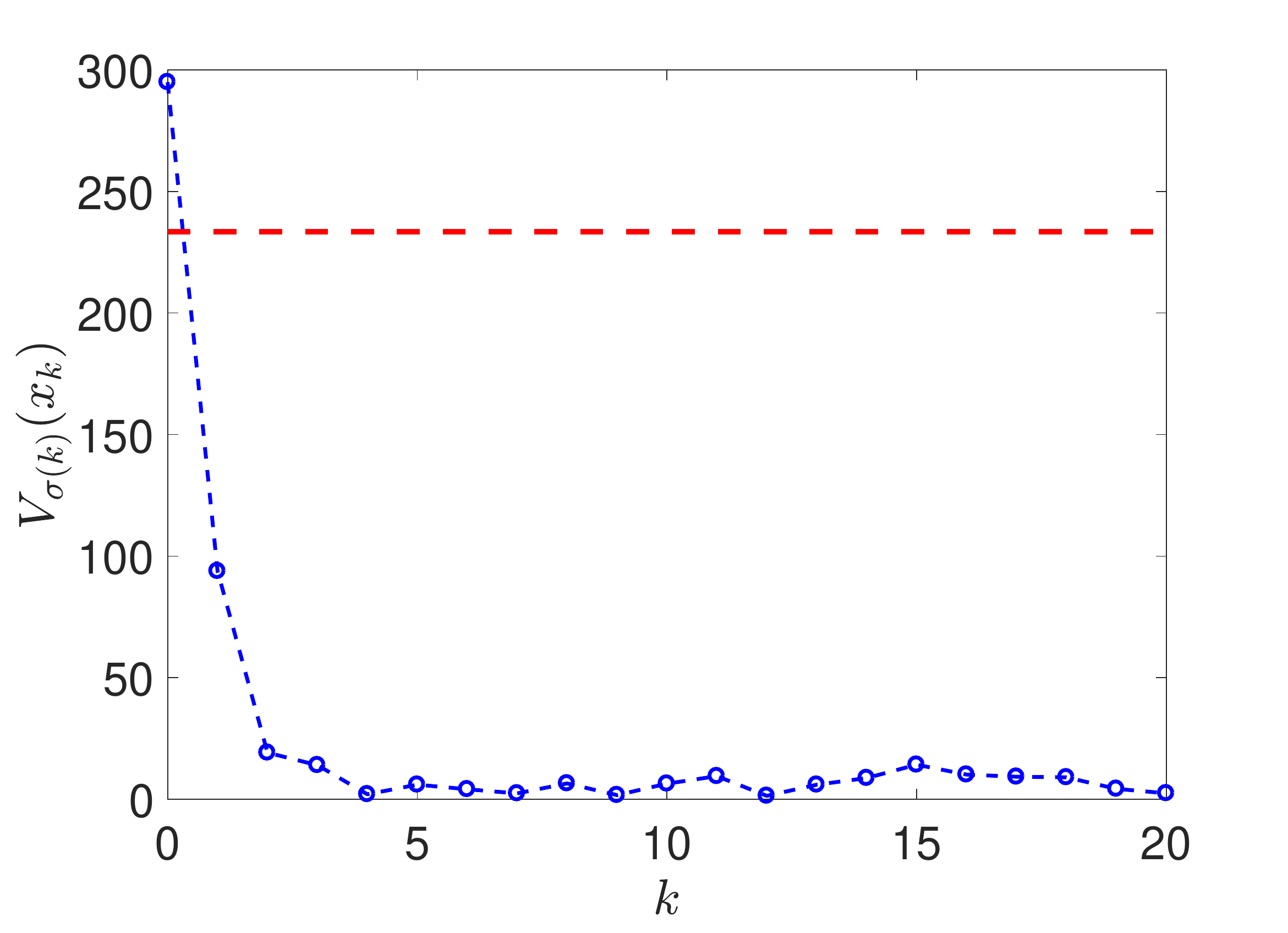}
\label{fig:V-sigma}
}
\centering
\subfigure[]
{
\includegraphics[width=0.4\textwidth]{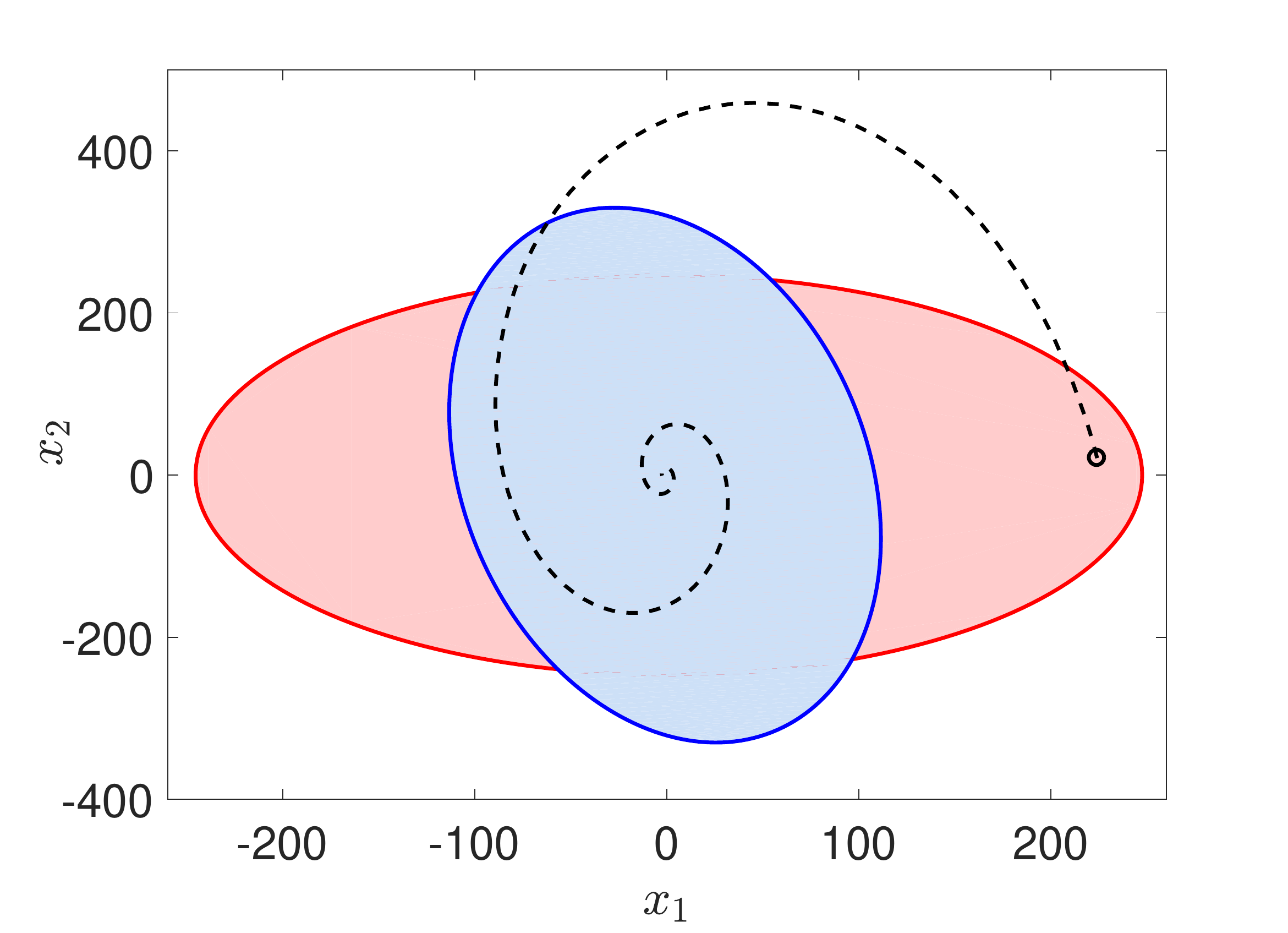}
\label{fig:counter}
}
\vskip -10pt
\caption{\textbf{(a)} Evolution of $V_{\sigma(k)}(x_k)$ (blue circles) over $k$ for switching among the members of the family \eqref{eq:disc-family-sim} while satisfying Theorem~\ref{thm:ISS-switch-discrete}. \textbf{(b)} State trajectory evolving under $d(t)\equiv 0$ that escapes $\mathcal{M}(c)$ before eventually getting trapped in it despite satisfying Theorem~\ref{thm:ISS-switch-continuous}. }
\vskip -10pt
\end{figure*}

Note that if $\underline{\alpha}_p(\cdot)$ and $\overline{\alpha}_p(\cdot)$ in Theorems~\ref{thm:ISS-switch-discrete} and \ref{thm:ISS-switch-continuous} are available for all $p\in\mathcal{P}$, then, following steps similar to those in the proof of Proposition~\ref{prop:mu-compute}, analytical bounds for $\omega(\kappa)$ can be obtained for these non-quadratic Lyapunov functions as well; however, we cannot obtain a general analytical bound for $\mu(\kappa)$. Note that the bound \eqref{eq:omega-bound} for $\omega(\kappa)$ has also been obtained in \cite{makarenkov2017dwell}.

\subsection{Numerical Examples}
\label{subsec:simulation}

\noindent \emph{Example 1.} Consider the family of linear discrete-time systems 
\begin{align}\label{eq:disc-family-sim}
x_{k+1} = A_p x_k + B_p d_k + C_p, & & p \in \mathcal{P} = \{1,2,3 \}  \enspace,
\end{align}
where $x_k\in\mathbb{R}^4$ is the state and $d_k\in\mathbb{R}$ is the disturbance. For each $p$ the eigenvalues of $A_p\in\mathbb{R}^{4\times 4}$ are within the unit disc, and $B_p\in\mathbb{R}^4$, $C_p\in\mathbb{R}^4$ are constant matrices; provided in Appendix~\ref{app:sim}. The fixed points for the individual subsystems are $x^*_1=(0.5,0,0,0)$, $x^*_2=(-0.5,0,0,0)$, and $x^*_3=(0,0.5,0,0)$. 

From \cite[Example~3.4]{jiang2001input} it is known that the Lyapunov function of an exponentially stable linear discrete-time system is also an ISS-Lyapunov function; hence, the quadratic ISS-Lyapunov functions can be found by solving the discrete Lyapunov equation on the 0-disturbance system. Using the analytical expressions in \cite[Example~3.4]{jiang2001input}, we compute $\lambda=0.68$ in \eqref{eq:V-ISS-2}. Using Proposition~\ref{prop:mu-compute} with $\kappa=10$ gives $\mu(\kappa)\leq 2.57$ and $\omega(\kappa) \leq 25.2$; it is remarked that since the state is in $\mathbb{R}^4$, estimating $\mu(\kappa)$ and $\omega(\kappa)$ without Proposition~\ref{prop:mu-compute} would be computationally intensive. With $\delta = 0.84$ we compute the lower bound on the average dwell time \eqref{eq:avg-dwell-lb-disc} as $N_{\rm a}\geq 4.47$. In Fig.~\ref{fig:V-sigma}, we provide the evolution of this switched system with initial state $x_0 = (5,10,-7.5,5)$ and a switching signal that satisfies \eqref{eq:avg-dwell-time-def} with $N_0=2$ and $N_{\rm a}=4.47$. The disturbance in this example satisfies $\|d\|_\infty = 10$ and is generated by sampling a uniform distribution on $[-10,10]$. The red line in Fig.~\ref{fig:V-sigma} is the computation of the level of the trapping compact set $\mathcal{M}(\bar{\omega})$ with $\bar{\omega}$ as in \eqref{eq:omega-bar-disc}. Note that $V_{\sigma(k)}(x_k)$ is eventually trapped below it.

\vspace{+0.05in}
\noindent \emph{Example 2.}  This example shows that the ultimately trapping set for zero disturbance $\mathcal{M}(c)$ is not invariant. Consider the family of continuous linear switched systems
\begin{align}\label{eq:cont-family-sim}
\dot{x} = A_p x + C_p, & & p \in \mathcal{P} = \{1,2\}  \enspace,
\end{align}
where $x\in\mathbb{R}^2$, and $A_p\in\mathbb{R}^{2\times 2}$, $C_p\in\mathbb{R}^2$ are constant matrices provided in Appendix~\ref{app:sim}. For each $p\in\{1,2\}$, $A_p$ is Hurwitz. The equilibria for the members of the family \eqref{eq:cont-family-sim} are $x_1^*=(1,0)$ and $x_2^*=(-1,0)$. A quadratic Lyapunov function for each subsystem is computed using the continuous Lyapunov equation. Choosing $\kappa=1$ and invoking Proposition~\ref{prop:mu-compute} we compute $\mu(\kappa) \leq 31.88$ and $\omega(\kappa) \leq  29.89$. With $\lambda=0.5$ in \eqref{eq:V-ISS-cont-p-ind} and choosing $\delta = 0.1$, we can compute $\overline{N}_{\rm a}=8.65$. From \eqref{eq:omega-bar-cont}, we have $c=\mu^{1+N_0}\omega$ in Theorem~\ref{thm:ISS-switch-continuous} for any $N_0\geq 1$. Choosing $N_0=1$, we get $c=3.064\times 10^4$. The set $\mathcal{M}(c)$ is plotted in Fig.~\ref{fig:counter} as the union of the red and blue ellipses. A solution starting from $x(0)=(224,21)$ evolves under the influence of the subsystem $p=2$ whose sublevel set is blue. Since we do not switch during the simulation we trivially satisfy the average dwell time constraint with $N_0=1$ and $N_{\rm a} \geq \overline{N}_{\rm a}$, hence we satisfy Theorem~\ref{thm:ISS-switch-continuous} with zero disturbance. Note that $x(0)$ marked by the circle in the plot lies within $\mathcal{M}(c)$, still the black state trajectory escapes the set $\mathcal{M}(c)$ before eventually getting trapped in it, demonstrating that the set is not invariant. 




\section{Conclusions}
\label{sec:conclusions}
In this paper we proposed a framework within which we can rigorously analyze the robustness under exogenous disturbances of continuous and discrete switched systems with multiple equilibria. It was shown that, if an average dwell-time constraint is satisfied by the switching signals, the solutions of the switched system are uniformly bounded and uniformly ultimately bounded in an explicitly characterizable compact set. Furthermore, the lower bound on the average dwell-time is analytically computable if the ISS-Lyapunov functions of the individual subsystems are quadratic. Although, our motivation for studying this class of systems arises from switching among motion primitives for robots under exogenous disturbances, the results of this paper are relevant to a much broader class of applications in which switching occurs among systems that do not share the same equilibrium point.

\appendices

\section{} 
\label{app:lemma-disc}

\begin{proof}[Proof of Lemma~\ref{lem:V-sigma}]
The statement of Lemma~\ref{lem:V-sigma} holds for an arbitrary initial time $\underline{k}$; to avoid cumbersome expressions, we prove the result for $\underline{k}=0$ noting that the same proof carries to the case of an arbitrary $\underline{k}$ by replacing $k$ with $k-\underline{k}$ in the expressions. We consider switching signals $\sigma : \mathbb{Z}_+ \to \mathcal{P}$ that satisfy Definition~\ref{def:average-dwell-time} for $N_0 \geq 1$ and $N_{\rm a} \geq \overline{N}_{\rm a}$, where $\overline{N}_{\rm a}$ is given by \eqref{eq:avg-dwell-lb-disc}. Let $\{k_1, k_2,...\}$ be a sequence of switching times for such signal.  
%
For notational compactness, define
\begin{align}
G_a^b(r):=
\begin{cases}
\sum_{j=0}^{b-a-1} r^j = \frac{1-r^{b-a}}{1-r} & \mathrm{if}~b>a\\
0 & \mathrm{if}~b=a
\end{cases} \enspace. \nonumber
\end{align}
where $a,b\in\mathbb{Z}_+$, $b\geq a$, and $0<r<1$. Further, we denote $N_{\sigma}(k,0)$ by $N_\sigma$ unless a different time window is specified.

Using \eqref{eq:V-ISS-2} over the interval $0\leq k < k_1$ until the first switching occurs, results in  
\begin{equation}\label{eq:ind-1}
V_{\sigma(k)}(x_k) \leq \lambda^k V_{\sigma(0)}(x_0) + G_0^k(\lambda)\hat{\alpha}(\|d\|_\infty) \enspace.
\end{equation}
Now, since $\mu \geq 1$ by \eqref{eq:mu_1} and $\lambda<\delta$, \eqref{eq:ind-1} results in
\begin{equation}\label{eq:ind-1-N1}
V_{\sigma(k)}(x_k) \leq \mu^{N_0}\delta^k V_{\sigma(0)}(x_0) + \frac{\mu^{N_0}}{1-\delta}\hat{\alpha}(\|d\|_\infty) \enspace,
\end{equation}
where we have used $G_0^k(\lambda) \leq G_0^k(\delta) \leq \frac{1}{1-\delta}$. Hence, \eqref{eq:ISS} holds for all $0 \leq k < k_1$, completing the proof if $N=1$.

Next, if $N \neq 1$ so that $x_{k_1} \notin \accentset{\circ}{\mathcal{M}}_{\sigma(k_1-1)}$, we can apply \eqref{eq:mu-bound-V} to relate the values at the switching state $x_{k_1}$ of the Lyapunov functions of the presently active system $\sigma(k_1)$ and of the formerly active system $\sigma(k_1-1)=\sigma(0)$. Hence, using \eqref{eq:ind-1} first to obtain the bound $V_{\sigma(k_1-1)}(x_{k_1})\leq \lambda^{k_1} V_{\sigma(0)}(x_0) + G_0^{k_1}(\lambda)\hat{\alpha}(\|d\|_\infty)$, we can then apply \eqref{eq:mu-bound-V} to obtain $V_{\sigma(k_1)}(x_{k_1})\leq\mu \lambda^{k_1} V_{\sigma(0)}(x_0) + \mu G_0^{k_1}(\lambda)\hat{\alpha}(\|d\|_\infty)$. This is used in \eqref{eq:V-ISS-2} to write the following bound for $k_1\leq k < k_2$, 

\vspace{-3mm}
\small
\begin{equation}\nonumber
V_{\sigma(k)}(x_k) \leq \mu\lambda^k V_{\sigma(0)}(x_0) + \Big(G_{k_1}^k(\lambda) + \mu \lambda^{k-k_1} G_0^{k_1}(\lambda)\Big)\hat{\alpha}(\|d\|_\infty) .
\end{equation}
\normalsize
Inductively repeating this process for $N_\sigma$ switches with $1\leq N_\sigma < N$ we have the following bound for $k_{N_\sigma}\leq k < k_{N_\sigma+1}$,
\begin{align}
& V_{\sigma(k)}(x_k) \leq \mu^{N_\sigma} \lambda^k V_{\sigma(0)}(x_0)  \label{eq:V-induction}\\
& + \bigg(G_{k_{N_\sigma}}^k(\lambda)+ \sum_{j=0}^{N_\sigma-1} \mu^{N_\sigma-j} \lambda^{k-k_{j+1}} G_{k_j}^{k_{j+1}}(\lambda) \bigg) \hat{\alpha}(\|d\|_\infty) \nonumber \enspace,
\end{align}
where $k_j=0$ for $j=0$. We treat the state- and disturbance-dependent terms in the upper bound of \eqref{eq:V-induction} separately. For the state-dependent term, recall that $\mu \geq 1$ by \eqref{eq:mu_1} and use \eqref{eq:avg-dwell-time-def} followed by $N_{\rm a}\geq \overline{N}_{\rm a}$ where $\overline{N}_{\rm a}$ satisfies \eqref{eq:avg-dwell-lb-disc} to get,
\begin{align}
\mu^{N_\sigma} \lambda^k V_{\sigma(0)}(x_0) & \leq \mu^{N_0} (\lambda \mu^{1/N_{\rm a}})^k V_{\sigma(0)}(x_0) \nonumber\\
& \leq \mu^{N_0} \delta^k V_{\sigma(0)}(x_0) \enspace.\label{eq:IC-part}
\end{align}
To proceed with the disturbance-dependent term, first note that $N_\sigma-j=N_\sigma(k,k_{j+1})$. Hence, using \eqref{eq:avg-dwell-time-def} on $N_\sigma(k,k_{j+1})$ followed by $N_{\rm a}\geq \overline{N}_{\rm a}$ with $\overline{N}_{\rm a}$ given by \eqref{eq:avg-dwell-lb-disc} results in
%
\begin{align}
\mu^{N_\sigma-j} & \leq \mu^{N_0} \mu^{(k-k_{j+1})\ln(\delta/\lambda)/\ln(\mu)} \nonumber\\
& \leq \mu^{N_0}(\delta/\lambda)^{k-k_{j+1}} \enspace. \label{eq:mu-j-bound}
\end{align}
Using \eqref{eq:mu-j-bound} in the summation in \eqref{eq:V-induction} gives
\small
\begin{align}
\sum_{j=0}^{N_\sigma-1} \lambda^{k-k_{j+1}} \mu^{N_\sigma-j} G_{k_j}^{k_{j+1}}(\lambda) & \leq \mu^{N_0}\sum_{j=0}^{N_\sigma-1} \delta^{k-k_{j+1}} G_{k_j}^{k_{j+1}}(\lambda) \nonumber \\
& \leq \mu^{N_0}\sum_{j=0}^{N_\sigma-1} \delta^{k-k_{j+1}} G_{k_j}^{k_{j+1}}(\delta) \label{eq:Gk-sum-1}
\end{align}
\normalsize
where the last inequality follows from the fact that $\lambda < \delta$, hence $G_{k_j}^{k_{j+1}}(\lambda) \leq G_{k_j}^{k_{j+1}}(\delta)$ with equality holding in the case when $k_{j+1} = k_j+1$. It can be easily verified that
\begin{align} \nonumber
\delta^{k-k_{j+1}} G_{k_j}^{k_{j+1}}(\delta) = \delta^{k-k_{j+1}} + \delta^{k-k_{j+1}+1} + ... + \delta^{k-k_j-1} \enspace,
\end{align}
which, on summing from $j=0$ to $j=N_\sigma-1$ and after some algebraic manipulation, results in
\begin{align}\label{eq:Gk-sum-2}
\sum_{j=0}^{N_\sigma-1}  \delta^{k-k_{j+1}} G_{k_j}^{k_{j+1}}(\delta)  &= \sum_{j=k-k_1}^{k-1}  \delta^j + \sum_{j=k-k_2}^{k-k_1-1} \delta^j + ... \nonumber \\
&+ \sum_{j=k-N_\sigma}^{k - k_{N_\sigma-1} - 1} \delta^j = \sum_{j=k-k_{N_\sigma}}^{k-1} \delta^j .
\end{align}
Using \eqref{eq:Gk-sum-2} in \eqref{eq:Gk-sum-1} gives
\begin{align}
\sum_{j=0}^{N_\sigma-1} \lambda^{k-k_{j+1}} \mu^{N_\sigma-j} G_{k_j}^{k_{j+1}}(\lambda) \leq \mu^{N_0}\sum_{j=k-k_{N_\sigma}}^{k-1} \delta^j \enspace. \label{eq:Gk-sum-3}
\end{align}
Additionally, as $\mu\geq 1$ by \eqref{eq:mu_1} and $\lambda<\delta$,
\begin{align}\label{eq:Gk-sum-4}
G_{k_{N_\sigma}}^k(\lambda) \leq \mu^{N_0} G_{k_{N_\sigma}}^k(\delta)  = \mu^{N_0}\sum_{j=0}^{k-k_{N_\sigma}-1} \delta^j \enspace.
\end{align}
Thus, using \eqref{eq:Gk-sum-3} and \eqref{eq:Gk-sum-4} on the disturbance dependent term of the upper bound in \eqref{eq:V-induction} gives
\begin{align}
 & \bigg(G_{k_{N_\sigma}}^k(\lambda)+ \sum_{j=0}^{N_\sigma-1} \lambda^{k-k_{j+1}} \mu^{N_\sigma-j} G_{k_j}^{k_{j+1}}(\lambda) \bigg) \hat{\alpha}(\|d\|_\infty) \nonumber\\
&\leq \mu^{N_0}\sum_{j=0}^{k-1} \delta^j \hat{\alpha}(\|d\|_\infty)  \leq \frac{\mu^{N_0}}{1-\delta} \hat{\alpha}(\|d\|_\infty) \enspace. \label{eq:gain-part}
\end{align}
Hence, upper bounding \eqref{eq:V-induction} with \eqref{eq:IC-part} and \eqref{eq:gain-part} gives \eqref{eq:ISS} for $k_{N_\sigma}\leq k<k_{N_\sigma + 1}$ for any $1\leq N_\sigma <N$, i.e., for all $k_1\leq k<k_N$. Further, by \eqref{eq:ind-1-N1}, \eqref{eq:ISS} holds for $0\leq k <k_1$. Hence, \eqref{eq:ISS} holds for all $0\leq k < k_N$.

Now we turn our attention to \eqref{eq:beta-alpha-bound-disc}. Using \eqref{eq:V-ISS-1} in \eqref{eq:ISS},
\begin{align}
\underline{\alpha}_{\sigma(k)}(\|x_k-x_{\sigma(k)}^*\|) & \leq \mu^{N_0} \delta^k \overline{\alpha}_{\sigma(0)}(\|x_0-x_{\sigma(0)}^*\|) \nonumber \\
& + \frac{\mu^{N_0}}{1-\delta} \hat{\alpha}(\|d\|_\infty) \enspace. \label{eq:beta-alpha-bound-1}
\end{align}
As $\underline{\alpha}_{\sigma(k)}^{-1}\in\mathcal{K}_\infty$ is monotonically increasing, by \eqref{eq:beta-alpha-bound-1}

\vspace{-3mm}
\small
\begin{align}
& \|x_k-x_{\sigma(k)}^*\| \nonumber \\
& \leq \underline{\alpha}_{\sigma(k)}^{-1} \bigg(\mu^{N_0} \delta^k \overline{\alpha}_{\sigma(0)}(\|x_0-x_{\sigma(0)}^*\|)   + \frac{\mu^{N_0}}{1-\delta} \hat{\alpha}(\|d\|_\infty) \bigg)  \nonumber \\
& \leq \underline{\alpha}_{\sigma(k)}^{-1} \bigg(2\mu^{N_0} \delta^k \overline{\alpha}_{\sigma(0)}(\|x_0-x_{\sigma(0)}^*\|) \bigg)  + \underline{\alpha}_{\sigma(k)}^{-1} \bigg( \frac{2\mu^{N_0}}{1-\delta} \hat{\alpha}(\|d\|_\infty) \bigg) \label{eq:beta-alpha-bound-2}
\end{align}
\normalsize
where the last inequality follows by \cite[Lemma~14]{sarkans2016input} with $\epsilon=1$. Observe that the first term in \eqref{eq:beta-alpha-bound-2} is in class $\mathcal{KL}$ while the second is in class $\mathcal{K}_\infty$. Let $s\in\mathbb{R}_+$ and define $\beta\in \mathcal{KL}$ as
\begin{align}\label{eq:beta-discrete}
\beta(s,k):=\max_{p,q\in\mathcal{P}} \underline{\alpha}_p^{-1} \big(2\mu^{N_0} \delta^k \overline{\alpha}_q(s) \big) \enspace.
\end{align}
Further define $\alpha\in\mathcal{K}_\infty$ as 
\begin{align}\label{eq:alpha-discrete}
\alpha(s):= \max_{p\in\mathcal{P}} \underline{\alpha}_p^{-1} \bigg( \frac{2\mu^{N_0}}{1-\delta} \hat{\alpha}(s) \bigg) \enspace.
\end{align}
Using \eqref{eq:beta-discrete} and \eqref{eq:alpha-discrete} in \eqref{eq:beta-alpha-bound-2} gives \eqref{eq:beta-alpha-bound-disc}.
\end{proof}

\section{}
\label{app:lemma-cont}

\begin{proof}[Proof of Lemma~\ref{lem:V-sigma-cont}]
We only provide a sketch due to the similarity with the proof of Lemma~\ref{lem:V-sigma}. The bound \eqref{eq:ISS-cont} follows from the observation that for any switching instant between $\underline{t}\leq t<t_N$, we have $x(t_n)\not\in\accentset{\circ}{\mathcal{M}}_{\sigma(t_n^-)}(\kappa)$, permitting the use of $\mu$ for this interval. Thus, we can use similar arguments as in the proof of \cite[Theorem~3.1]{vu2007input} to obtain \eqref{eq:ISS-cont}. With the availability of \eqref{eq:ISS-cont}, we can follow steps identical to the proof of \eqref{eq:beta-alpha-bound-disc} to obtain \eqref{eq:beta-alpha-bound-cont}. The only difference is in the  class $\mathcal{KL}$ and class $\mathcal{K}_\infty$ functions involved in the estimates.
\end{proof}

\section{}
\label{app:sim}

We provide the system \eqref{eq:disc-family-sim} of Example 1 of Section~\ref{subsec:simulation}.

\small
\begin{equation}\nonumber
A_1 = 
\begin{bmatrix}
0.5233 & -0.0041 & -0.0609 & -0.0081 \\
-0.0041 & 0.4762 & -0.0526 & -0.0476 \\
-0.0609 & -0.0526 &  0.4957 & -0.1052 \\
-0.0081 &  -0.0476  & -0.1052  &  0.4048
\end{bmatrix} \enspace,
\end{equation}
\begin{equation} \nonumber
A_2 = \begin{bmatrix}
0.3133  & -0.0324  &  0.0062 &   0.0082 \\
-0.0324  &  0.3787  & -0.0150  & -0.0200 \\
    0.0062 &  -0.0150  &  0.3029  &  0.0038 \\
    0.0082  & -0.0200  &  0.0038  &  0.3051 
\end{bmatrix} \enspace,
\end{equation}
\begin{equation} \nonumber
A_3 = \begin{bmatrix}
0.4100  &  0.0003  &  0.0278 &  -0.0360 \\
    0.0003  &  0.3944  &  0.0001  & -0.0224 \\
    0.0278  &  0.0001  &  0.3629  & -0.0167 \\
   -0.0360  & -0.0224 &  -0.0167  &  0.3827 
\end{bmatrix} \enspace,
\end{equation}
\begin{equation}\nonumber
B_1=B_2 = B_3 = 
\begin{bmatrix}
0.1 & 0.1 & 0.1 & 0.1
\end{bmatrix}^{\rm T} \enspace,
\end{equation}
\begin{equation}\nonumber
C_1 = 
\begin{bmatrix}
0.2383 & 0.0020 & 0.0304 & 0.0041
\end{bmatrix}^{\rm T} \enspace,
\end{equation}
\begin{equation}\nonumber
C_2 = 
\begin{bmatrix}
-0.3433 & -0.0162 & 0.0031 & 0.0041
\end{bmatrix}^{\rm T} \enspace,
\end{equation}
\begin{equation}\nonumber
C_3 = 
\begin{bmatrix}
-0.0001 & 0.3028 & -0.0001 & 0.0112
\end{bmatrix}^{\rm T} \enspace.
\end{equation}

\normalsize
We provide the system \eqref{eq:cont-family-sim} of Example 2 of Section~\ref{subsec:simulation}.

\small
\begin{align}\nonumber
A_1 = 
\begin{bmatrix}
-1  &  -1 \\
 1  &  -1
\end{bmatrix},
C_1 = 
\begin{bmatrix}
1 \\ -1
\end{bmatrix}, 
A_2 = \begin{bmatrix}
-1  &  -1 \\
 10  &  -1
\end{bmatrix},
C_2 = 
\begin{bmatrix}
-1  \\ 10
\end{bmatrix} .
\end{align}

\ifCLASSOPTIONcaptionsoff
  \newpage
\fi



%

%
\bibliographystyle{IEEEtran}
\bibliography{Avg_Dwell_Time}

\begin{thebibliography}{10}
\providecommand{\url}[1]{#1}
\csname url@samestyle\endcsname
\providecommand{\newblock}{\relax}
\providecommand{\bibinfo}[2]{#2}
\providecommand{\BIBentrySTDinterwordspacing}{\spaceskip=0pt\relax}
\providecommand{\BIBentryALTinterwordstretchfactor}{4}
\providecommand{\BIBentryALTinterwordspacing}{\spaceskip=\fontdimen2\font plus
\BIBentryALTinterwordstretchfactor\fontdimen3\font minus
  \fontdimen4\font\relax}
\providecommand{\BIBforeignlanguage}[2]{{%
\expandafter\ifx\csname l@#1\endcsname\relax
\typeout{** WARNING: IEEEtran.bst: No hyphenation pattern has been}%
\typeout{** loaded for the language `#1'. Using the pattern for}%
\typeout{** the default language instead.}%
\else
\language=\csname l@#1\endcsname
\fi
#2}}
\providecommand{\BIBdecl}{\relax}
\BIBdecl

\bibitem{vasca2012dynamics}
F.~Vasca and L.~Iannelli, \emph{Dynamics and control of switched electronic
  systems: {A}dvanced perspectives for modeling, simulation and control of
  power converters}.\hskip 1em plus 0.5em minus 0.4em\relax Springer, 2012.

\bibitem{johansen2003gain}
T.~A. Johansen, I.~Petersen, J.~Kalkkuhl, and J.~Ludemann, ``Gain-scheduled
  wheel slip control in automotive brake systems,'' \emph{IEEE Tr. on Control
  Systems Technology}, vol.~11, no.~6, pp. 799--811, 2003.

\bibitem{aguiar2007trajectory}
A.~P. Aguiar and J.~P. Hespanha, ``Trajectory-tracking and path-following of
  underactuated autonomous vehicles with parametric modeling uncertainty,''
  \emph{IEEE Tr. on Automatic Control}, vol.~52, no.~8, pp. 1362--1379, 2007.

\bibitem{tomlin1996hybrid}
C.~Tomlin, G.~Pappas, J.~Lygeros, D.~Godbole, S.~Sastry, and G.~Meyer, ``Hybrid
  control in air traffic management system,'' \emph{IFAC Proceedings Volumes},
  vol.~29, no.~1, pp. 5512--5517, 1996.

\bibitem{narendra1994common}
K.~S. Narendra and J.~Balakrishnan, ``A common lyapunov function for stable
  {LTI} systems with commuting {A}-matrices,'' \emph{IEEE Tr. on {A}utomatic
  {C}ontrol}, vol.~39, no.~12, pp. 2469--2471, 1994.

\bibitem{branicky1998multiple}
M.~S. Branicky, ``Multiple lyapunov functions and other analysis tools for
  switched and hybrid systems,'' \emph{IEEE Tr. on {A}utomatic {C}ontrol},
  vol.~43, no.~4, pp. 475--482, 1998.

\bibitem{hespanha1999stability}
J.~P. Hespanha and A.~S. Morse, ``Stability of switched systems with average
  dwell-time,'' in \emph{Proc. of IEEE Conf. on Decision and Control}, vol.~3,
  1999, pp. 2655--2660.

\bibitem{vu2007input}
L.~Vu, D.~Chatterjee, and D.~Liberzon, ``Input-to-state stability of switched
  systems and switching adaptive control,'' \emph{Automatica}, vol.~43, no.~4,
  pp. 639--646, 2007.

\bibitem{lin2009stability}
H.~Lin and P.~J. Antsaklis, ``Stability and stabilizability of switched linear
  systems: a survey of recent results,'' \emph{IEEE Tr. on Automatic control},
  vol.~54, no.~2, pp. 308--322, 2009.

\bibitem{liberzon2003switching}
D.~Liberzon, \emph{Switching in {S}ystems and {C}ontrol}.\hskip 1em plus 0.5em
  minus 0.4em\relax Birkh\"{a}user, 2003.

\bibitem{motahar2016composing}
M.~S. Motahar, S.~Veer, and I.~Poulakakis, ``Composing limit cycles for motion
  planning of 3{D} bipedal walkers,'' in \emph{Proc. of IEEE Conf. on Decision
  and Control}, 2016, pp. 6368--6374.

\bibitem{gregg-planning2012}
R.~D. Gregg, A.~K. Tilton, S.~Candido, T.~Bretl, and M.~W. Spong, ``Control and
  planning of 3-{D} dynamic walking with asymptotically stable gait
  primitives,'' \emph{IEEE Tr. on Robotics}, vol.~28, no.~6, pp. 1415--1423,
  2012.

\bibitem{dorothy2016switched}
M.~Dorothy and S.-J. Chung, ``Switched systems with multiple invariant sets,''
  \emph{Systems \& Control Letters}, vol.~96, pp. 103--109, 2016.

\bibitem{figueredo2014switching}
L.~Figueredo, B.~V. Adorno, J.~Y. Ishihara, and G.~Borges, ``Switching strategy
  for flexible task execution using the cooperative dual task-space
  framework,'' in \emph{Proc. of IEEE/RSJ Int. Conf. on Intelligent Robots and
  Systems}, 2014, pp. 1703--1709.

\bibitem{alpcan2004hybrid}
T.~Alpcan and T.~Ba{\c{s}}ar, ``A hybrid systems model for power control in
  multicell wireless data networks,'' \emph{Performance Evaluation}, vol.~57,
  no.~4, pp. 477--495, 2004.

\bibitem{makarenkov2017dwell}
O.~Makarenkov and A.~Phung, ``Dwell time for local stability of switched
  systems with application to non-spiking neuron models,'' \emph{arXiv preprint
  arXiv:1712.07517}, 2017.

\bibitem{basar2010}
T.~Alpcan and T.~Basar, ``A stability result for switched systems with multiple
  equilibria,'' \emph{Dynamics of Continuous, Discrete and Impulsive Systems
  Series A: Mathematical Analysis}, vol.~17, pp. 949--958, 2010.

\bibitem{blanchini2010modal}
F.~Blanchini, D.~Casagrande, and S.~Miani, ``Modal and transition dwell time
  computation in switching systems: {A} set-theoretic approach,''
  \emph{Automatica}, vol.~46, no.~9, pp. 1477--1482, 2010.

\bibitem{kuiava2013practical}
R.~Kuiava, R.~A. Ramos, H.~R. Pota, and L.~F. Alberto, ``Practical stability of
  switched systems without a common equilibria and governed by a time-dependent
  switching signal,'' \emph{European J. of Control}, vol.~19, no.~3, pp.
  206--213, 2013.

\bibitem{burridge1999sequential}
R.~R. Burridge, A.~A. Rizzi, and D.~E. Koditschek, ``Sequential composition of
  dynamically dexterous robot behaviors,'' \emph{Int. J. of Robotics Research},
  vol.~18, no.~6, pp. 534--555, 1999.

\bibitem{tedrake2010}
R.~Tedrake, I.~R. Manchester, M.~Tobenkin, and J.~W. Roberts, ``{LQR}-trees:
  Feedback motion planning via sums-of-squares verification,'' \emph{Int. J. of
  Robotics Research}, vol.~29, no.~8, pp. 1038--1052, 2010.

\bibitem{QuCao2015}
Q.~Cao, A.~T. {van Rijn}, and I.~Poulakakis, ``On the control of gait
  transitions in quadrupedal running,'' in \emph{Proc. of IEEE/RSJ Int. Conf.
  on Intelligent Robots and Systems}, Sep. 2015, pp. 5136 -- 5141.

\bibitem{Cao2016ALR}
Q.~Cao and I.~Poulakakis, ``Quadrupedal running with a flexible torso:
  {C}ontrol and speed transitions with sums-of-squares verification,''
  \emph{Artificial Life and Robotics}, vol.~21, no.~4, pp. 384--392, 2016.

\bibitem{veer2017supervisory}
S.~Veer, M.~S. Motahar, and I.~Poulakakis, ``Adaptation of limit-cycle walkers
  for collaborative tasks: {A} supervisory switching control approach,'' in
  \emph{Proc. of IEEE/RSJ Int. Conf. on Intelligent Robots and Systems}, 2017,
  pp. 5840--5845.

\bibitem{quan2016dynamic}
Q.~Nguyen, X.~Da, J.~Grizzle, and K.~Sreenath, ``Dynamic walking on stepping
  stones with gait library and control barrier functions,'' in \emph{Proc. of
  Int. Workshop On the Algorithmic Foundations of Robotics}, 2016.

\bibitem{saglam2013switching}
C.~O. Saglam and K.~Byl, ``Switching policies for metastable walking,'' in
  \emph{Proc. of IEEE Conf. on Decision and Control}, 2013, pp. 977--983.

\bibitem{da2016first}
X.~Da, R.~Hartley, and J.~W. Grizzle, ``First steps toward supervised learning
  for underactuated bipedal robot locomotion, with outdoor experiments on the
  wave field,'' in \emph{Proc. of IEEE Int. Conf. on Robotics and Automation},
  2017, pp. 3476--3483.

\bibitem{veer2017poincare}
S.~Veer, Rakesh, and I.~Poulakakis, ``Input-to-state stability of periodic
  orbits of systems with impulse effects via {P}oincar\'e analysis,''
  \emph{arXiv preprint arXiv:1712.03291}, 2018.

\bibitem{veer2017driftless}
S.~Veer, M.~S. Motahar, and I.~Poulakakis, ``Almost driftless navigation of 3d
  limit-cycle walking bipeds,'' in \emph{Proc. of IEEE/RSJ Int. Conf. on
  Intelligent Robots and Systems}, 2017, pp. 5025--5030.

\bibitem{veer2017continuum}
------, ``Generation of and switching among limit-cycle bipedal walking
  gaits,'' in \emph{Proc. of IEEE Conf. on Decision and Control}, 2017, pp.
  5827--5832.

\bibitem{bhounsule2018switching}
P.~A. Bhounsule, A.~Zamani, and J.~Pusey, ``Switching between limit cycles in a
  model of running using exponentially stabilizing discrete control lyapunov
  function,'' in \emph{Proc. of the American Control Conference}, 2018, pp.
  3714--3719.

\bibitem{majumdar2017funnel}
A.~Majumdar and R.~Tedrake, ``Funnel libraries for real-time robust feedback
  motion planning,'' \emph{Int. J. of Robotics Research}, vol.~36, no.~8, pp.
  947--982, 2017.

\bibitem{ramezani2017describing}
A.~Ramezani, S.~U. Ahmed, J.~Hoff, S.-J. Chung, and S.~Hutchinson, ``Describing
  robotic bat flight with stable periodic orbits,'' in \emph{Biomimetic and
  Biohybrid Systems. Living Machines 2017}, M.~Mangan, M.~Cutkosky, A.~Mura,
  P.~Verschure, T.~Prescott, and N.~Lepora, Eds.\hskip 1em plus 0.5em minus
  0.4em\relax Springer, 2017, pp. 394--405.

\bibitem{liljeback2011controllability}
P.~Liljeback, K.~Y. Pettersen, {\O}.~Stavdahl, and J.~T. Gravdahl,
  ``Controllability and stability analysis of planar snake robot locomotion,''
  \emph{IEEE Tr. on Automatic Control}, vol.~56, no.~6, pp. 1365--1380, 2011.

\bibitem{nagarajan2014ballbot}
U.~Nagarajan, G.~Kantor, and R.~Hollis, ``The ballbot: An omnidirectional
  balancing mobile robot,'' \emph{Int. J. of Robotics Research}, vol.~33,
  no.~6, pp. 917--930, 2014.

\bibitem{khalil2002nonlinear}
H.~K. Khalil, \emph{Nonlinear {S}ystems}, 3rd~ed.\hskip 1em plus 0.5em minus
  0.4em\relax Upper Saddle River, New Jersey: Prentice hall, 2002.

\bibitem{jiang2001input}
Z.-P. Jiang and Y.~Wang, ``Input-to-state stability for discrete-time nonlinear
  systems,'' \emph{Automatica}, vol.~37, no.~6, pp. 857--869, 2001.

\bibitem{sontag1996new}
E.~D. Sontag and Y.~Wang, ``New characterizations of input-to-state
  stability,'' \emph{IEEE Tr. on Automatic Control}, vol.~41, no.~9, pp.
  1283--1294, 1996.

\bibitem{rudin1964principles}
W.~Rudin, \emph{Principles of mathematical analysis}.\hskip 1em plus 0.5em
  minus 0.4em\relax McGraw-Hill New York, 1964, vol.~3.

\bibitem{sarkans2016input}
E.~Sarkans and H.~Logemann, ``Input-to-state stability of discrete-time {L}ur'e
  systems,'' \emph{SIAM Journal on Control and Optimization}, vol.~54, no.~3,
  pp. 1739--1768, 2016.

\end{thebibliography}
%

\vspace{-1cm}
\begin{IEEEbiography}[{\includegraphics[width=1in,height=1.25in,clip,keepaspectratio]{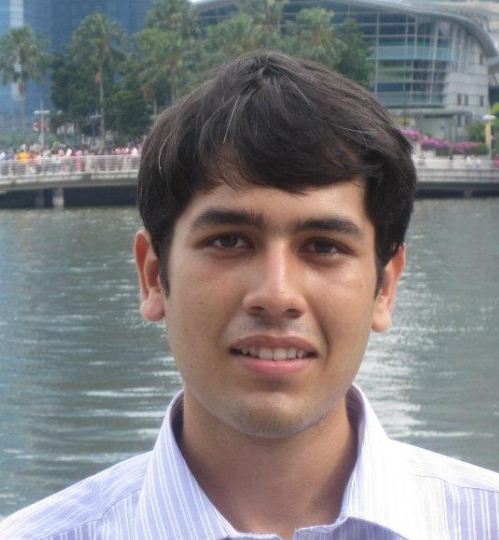}}]{Sushant Veer}
is a PhD Candidate in the Department of Mechanical Engineering at the University of Delaware. He received his B.Tech in Mechanical Engineering from the Indian Institute of Technology Madras (IIT-M) in 2013. His research interests lie in the control of complex dynamical systems with application to dynamically stable robots. 
\end{IEEEbiography}

\vspace{-1 cm}

\begin{IEEEbiography}[{\includegraphics[width=1in,height=1.25in,clip,keepaspectratio]{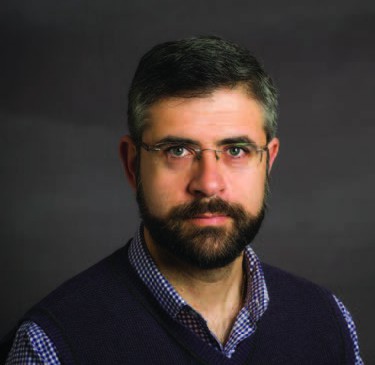}}]{Ioannis Poulakakis}
received his Ph.D. in Electrical Engineering (Systems) from the University of Michigan, MI, in 2009. From 2009 to 2010 he was a post-doctoral researcher with the Department of Mechanical and Aerospace Engineering at Princeton University, NJ. Since September 2010 he has been with the Department of Mechanical Engineering at the University of Delaware, where he is currently an Associate Professor. His research interests lie in the area of dynamics and control with applications to robotic systems, particularly dynamically dexterous legged robots. Dr. Poulakakis received the NSF CAREER Award in 2014. 
\end{IEEEbiography}





\end{document}